\documentclass[12pt]{article}
\usepackage{amsmath,amssymb,amsthm,graphicx,mathrsfs,color,fancyhdr,fancybox,setspace}
\usepackage{placeins}
\usepackage{amsfonts}
\usepackage[ruled,vlined,longend]{algorithm2e}

\newtheorem{theorem}{Theorem}[section]

\newtheorem{conjecture}[theorem]{Conjecture}
\newtheorem{definition}[theorem]{Definition}

\newtheorem{proposition}[theorem]{Proposition}
\newtheorem{lemma}[theorem]{Lemma}
\newtheorem{corollary}[theorem]{Corollary}
\newtheorem{claim}{Claim}

\newcommand{\qqed}{\hfill $\blacksquare$}

\usepackage[symbol]{footmisc}

\begin{document}
\thispagestyle{empty}
\begin{center}
\Large
 {\bf On the edge-biclique graph and the iterated edge-biclique operator}
\vspace*{1cm}

\large
\renewcommand{\thefootnote}{\fnsymbol{footnote}}
Leandro Montero\footnote{Corresponding author.} \\
\normalsize
KLaIM team, L@bisen, AIDE Lab., Yncrea Ouest\\33 Q, Chemin du Champ de Man\oe{}uvres\\44470 Carquefou, France\\
\texttt{lpmontero@gmail.com}\\

\vspace*{1cm}

\large
Sylvain Legay\\
 \normalsize
France\\
\vspace*{.5cm}

ABSTRACT
\end{center}

\small A biclique of a graph $G$ is a maximal induced complete bipartite subgraph of $G$. The edge-biclique graph of $G$, $KB_e(G)$, is the 
edge-intersection graph of the bicliques of $G$. A graph $G$ diverges (resp. converges or is periodic) under an operator $H$ whenever 
$\lim_{k \rightarrow \infty}|V(H^k(G))|=\infty$ (resp. $\lim_{k \rightarrow \infty}H^k(G)=H^m(G)$ for some $m$ or $H^k(G)=H^{k+s}(G)$ 
for some $k$ and $s \geq 2$). The iterated edge-biclique graph of $G$, $KB_e^k(G)$, is the graph obtained by applying the edge-biclique operator 
$k$ successive times to $G$. In this paper, we first study the connectivity relation between $G$ and $KB_e(G)$. Next, we study the iterated edge-biclique operator $KB_e$. 
In particular, we give sufficient conditions for a graph to be convergent or divergent under the operator $KB_e$, we characterize the
behavior of \textit{burgeon graphs} and we propose some general conjectures on the subject.

\normalsize
\vspace*{1cm}
\newpage

\section{Introduction}

Intersection graphs of certain special subgraphs of a general graph have been studied 
extensively. We can mention line graphs (intersection graphs of the edges of a graph), interval graphs (intersection graphs
of a family of subpaths of a path), and in particular, clique graphs (intersection graphs of the family of all maximal cliques of a graph)~\cite{BoothLuekerJCSS1976,BrandstadtLeSpinrad1999,EscalanteAMSUH1973,FulkersonGrossPJM1965,GavrilJCTSB1974,LehotJA1974,McKeeMcMorris1999}.

The \textit{clique graph} of $G$ is denoted by $K(G)$.
Clique graphs were introduced by Hamelink in \cite{HamelinkJCT1968} and characterized in \cite{RobertsSpencerJCTSB1971}. It was proved in \cite{Alc'onFariaFigueiredoGutierrez2006} that the 
clique graph recognition problem is NP-Complete.

The clique graph can be thought as an operator from \textit{Graphs} into \textit{Graphs}. The \textit{iterated clique graph} $K^k(G)$ 
 is the graph obtained by applying the clique operator $k$ successive times. It was introduced by Hedetniemi and Slater in \cite{HedetniemiSlater1972}. 
Much work has been done in the field of the iterated clique operator, looking at the possible different behaviors. The goal is to decide
 whether a given graph converges, diverges, or is periodic under the clique operator when $k$ grows to infinity. This question remains open 
for the general case, moreover, it is not known if it is computable. However, partial characterizations have been given for convergent, divergent
 and periodic graphs, restricted to some classes of graphs. Some of them lead to polynomial time algorithms to solve the problem. 

For the clique-Helly graph class, graphs which are convergent to the trivial graph have been characterized in~\cite{BandeltPrisnerJCTSB1991}.
 Cographs, $P_4$-tidy graphs, and circular-arc graphs  are examples of classes where the different behaviors were also 
characterized~\cite{MelloMorganaLiveraniDAM2006,LarrionMelloMorganaNeumann-LaraPizanaDM2004}. On the other hand, divergent graphs were considered. 
For example, in \cite{Neumann1981}, families of divergent graphs are given. Periodic graphs were studied in
 \cite{EscalanteAMSUH1973,LarrionNeumann-LaraPizanaDM2002}. It has been proved that for every integer $i$, there are graphs with
 period $i$ and graphs which converge in $i$ steps. More results about iterated clique graphs can be found in~\cite{Frias-ArmentaNeumann-LaraPizanaDM2004,LarrionNeumann-LaraGC1997,LarrionNeumann-LaraDM1999,LarrionNeumann-LaraDM2000,LarrionPizanaVillarroel-FloresDM2008,PizanaDM2003}.

A \textit{biclique} is a maximal induced complete bipartite subgraph. 
Bicliques have applications in various fields, for example biology: protein-protein interaction networks~\cite{Bu01052003}, 
social networks: web community discovery~\cite{Kumar}, genetics~\cite{Atluri}, medicine~\cite{Niranjan}, information theory~\cite{Haemers200156}, etc. 
More applications (including some of these) can be found in~\cite{blablamec}.
The \textit{biclique graph} of a graph $G$, denoted by $KB(G)$, is the intersection graph of the family of all bicliques of $G$. 
It was defined and characterized in~\cite{GroshausSzwarcfiterJGT2010}. However no polynomial time algorithm is known for recognizing biclique graphs. 
As for clique graphs, the biclique graph construction can be viewed as an operator between the class of graphs. 

The \textit{iterated biclique graph} $KB^k(G)$ is the graph obtained by applying to $G$ the biclique operator $k$ times iteratively. It was introduced in~\cite{marinayo} and all possible behaviors were characterized. It was proven that a graph is either divergent or convergent, but never periodic (with period bigger than $1$). Also, general characterizations for 
convergent and divergent graphs were given. 
These results were based on the fact that if a graph $G$ contains a clique of size at least $5$, 
then $KB(G)$ or $KB^2(G)$ contains a clique of larger size. Therefore, in that case $G$ diverges. Similarly if $G$ contains the $gem$ or the $rocket$ 
graphs as an induced subgraph, then $KB(G)$ 
contains a clique of size $5$, and again $G$ diverges. Otherwise it was shown that after removing false-twin vertices of $KB(G)$, the resulting graph is a clique on at most $4$ vertices, in which case $G$ converges. Moreover, it was proved that if a graph $G$ converges, it converges to the graphs $K_1$ or $K_3$, and it does so in at most $3$ steps. 
These characterizations led to an $O(n^4)$ time algorithm (later improved to $O(n+m)$ time~\cite{algorapide}) for recognizing convergent or divergent graphs under the biclique operator.

The \textit{edge-biclique graph} of a graph $G$, denoted by $KB_e(G)$, is the edge-intersection graph of the family of all bicliques of $G$. We recall that edge-intersection means that 
$KB_e(G)$ has a vertex for each biclique of $G$ and two vertices are adjacent in $KB_e(G)$ if their corresponding bicliques in $G$ share an edge (and not just a vertex as in $KB(G)$).
The edge-biclique graph $KB_e(G)$ was defined in~\cite{tesismarina} and studied in~\cite{pavol}, however there is no characterization so far to recognize edge-biclique graphs. 

In this work we study edge-biclique graphs not only because of their mathematical interest but also because in real-life problems, bicliques often represent the relation between two types of entities (each partition of the biclique) therefore it would make sense to study when two objects (bicliques) share a common relationship (an edge) more than just an entity (a vertex). 

We first study the relation between $G$ and $KB_e(G)$ in terms of connectivity and we present a polynomial time algorithm to decide if $KB_e(G)$ is connected or not.
In the rest of the paper, we define and focus on the \textit{iterated edge-biclique graph}, denoted by $KB_e^k(G)$, that is, the graph obtained by applying to $G$ the edge-biclique operator $k$ 
times iteratively. We give some non-trivial sufficient conditions for a graph to be convergent or divergent under the $KB_e$ operator
that are based on induced substructures. Later, we study \textit{burgeon graphs} and its relation
with line graphs and edge-biclique graphs\footnote{\textit{Burgeon graphs} have been studied under the name of \textit{inflated graphs} mainly considering the domination problem~\cite{burgeon1,burgeon3,burgeon2}.}. 
We also characterize its behavior under the $KB_e$ operator.
To finish, we propose some conjectures that would help to fully characterize the behavior of any graph under the $KB_e$ operator.

This work is organized as follows. In Section 2 the necessary notation is given. In Section 3 we give connectivity results of $KB_e(G)$.
In Section 4 and Section 5 we present some results about convergent and divergent graphs, respectively. In Section 6, we study \textit{burgeon graphs}.
Finally, in Section 7 we state some general conjectures on the subject.

\section{Preliminaries}

Along the paper we restrict to undirected simple graphs. Let $G=(V,E)$ be a graph with vertex set $V(G)$ and edge set $E(G)$, and 
let $n=|V(G)|$ and $m=|E(G)|$. A \textit{subgraph} $G'$ of $G$ is a graph $G'=(V',E')$, where $V'\subseteq V$ and 
$E'\subseteq E$ such that all endpoints of the edges of $E'$ are in $V'$. When $E'$ has all the edges of $E$ whose endpoints belong to the vertex subset $V'$, we say that $V'$ \textit{induces} the subgraph $G'=(V',E')$, that is, $G'$ is an \textit{induced subgraph} of $G$.
Also, let $G[V']$ denote the induced subgraph of $G$ by the set $V'$.
A graph $G=(V,E)$ is \textit{bipartite} when there exist sets $U$ and $W$ such that $V= U \cup W$, $U \cap W = \emptyset$, $U \neq \emptyset$, $W \neq \emptyset$ and $E \subseteq U \times W$. Say that $G$ is a 
\textit{complete graph} when every possible edge belongs to $E$. A complete graph on $n$ vertices is denoted $K_{n}$. A bipartite graph is \textit{complete bipartite} when
every vertex of the first set is connected to every vertex of the second set. A complete bipartite graph on $p$ vertices in one set and $q$ vertices in the other is denoted $K_{p,q}$.
A \textit{clique} of $G$ is a maximal complete induced subgraph, while a \textit{biclique} is a maximal induced complete bipartite subgraph of $G$. 
The \textit{open neighborhood} of a vertex $v \in V(G)$, denoted $N(v)$, is the set of vertices adjacent to $v$. 
The \textit{closed neighborhood} of a vertex $v \in V(G)$, denoted $N[v]$, is the set $N(v) \cup \{v\}$.
Given a vertex $v \in V(G)$ and set of vertices $S \subseteq V(G)$, we denote by $N_S(v)$, to the neigborhood of the vertex $v$ restricted to the set $S$.
Given a set of vertices $S \subseteq V(G)$, $\overline{S}$ denotes the set $V(G) - S$.
The \textit{degree} of a vertex $v$, denoted by $d(v)$, is defined as $d(v) = |N(v)|$. A \textit{path} (\textit{cycle}) on $k$ vertices ($k\geq 3$), denoted by $P_k$ ($C_{k}$), is a sequence of vertices 
$v_{1},v_{2},...,v_{k} \in G$ such that $v_{i} \neq v_{j}$ for all $1 \leq i \neq j \leq k$ and $v_{i}$ is adjacent to $v_{i+1}$ 
for all $1 \leq i \leq k-1$ (and $v_k$ is adjacent to $v_1$). A graph is \textit{connected} if there exists a path between each pair of vertices.
The \textit{girth} of $G$ is the length of a shortest induced cycle in the graph. Unless stated otherwise, we assume that all graphs of this paper are connected.

Given a family of sets $\mathcal{H}$, the \textit{intersection graph} of $\mathcal{H}$ is a graph that has the members of 
$\mathcal{H}$ as vertices, and there is an edge between two sets $E,F\in\mathcal{H}$ when $E$ and $F$ have non-empty intersection.

A graph $G$ is an \textit{intersection graph} if there exists a family of sets $\mathcal{H}$ such that $G$ is the intersection 
graph of $\mathcal{H}$. We remark that any graph is an intersection graph \cite{Szpilrajn-MarczewskiFM1945}. 

Let $H$ be any graph operator and let $G$ be a graph. The \textit{iterated graph under the operator} $H$ is defined iteratively as follows: 
$H^{0}(G)=G$ and for $k\geq 1$, $H^{k}(G)=H^{k-1}(H(G))$. We say that $G$ \textit{diverges} 
(resp. \textit{converges} or \textit{is periodic}) under the operator $H$ whenever $\lim_{k \rightarrow \infty}|V(H^{k}(G))|=\infty$ 
(resp. $\lim_{k \rightarrow \infty}H^{k}(G)=H^{m}(G)$ for some $m$ or $H^{k}(G)=H^{k+s}(G)$ for some $k$ and $s \geq 2$). 
The study of the \textit{behavior of a graph} $G$ under the operator $H$ consists of deciding if $G$ converges, diverges or is periodic under $H$.

We assume that the empty graph is convergent under the operator $KB_e$, as it is obtained by appyling the edge-biclique operator to
a graph that does not contain any bicliques.

\section{Connectivity}

In this section we will study the connectivity relation between $G$ and $KB_e(G)$. In comparison to the biclique graph $KB(G)$, it was shown in~\cite{tesislea,arxivdist} that
$G$ is connected if and only if $KB(G)$ is connected. This result is no longer true for edge-biclique graphs. For example, just observe that $KB_e(K_n)$ consists of $\frac{n(n-1)}{2}$ isolated vertices, i.e., it is disconnected.

The main result of this section is the following Theorem that characterizes when $KB_e(G)$ is connected.

\begin{theorem}\label{conexo}
Let $G$ be a connected graph. $KB_e(G)$ is connected if and only if there is no subset of vertices $S \subsetneq V(G)$ such that for every $v,w \in S$ we
have $N_{\overline{S}}(v)=N_{\overline{S}}(w)$, and $|E(G[S])| \geq 1$, that is, the subgraph induced by $S$ has at least one edge.
\end{theorem}
\begin{proof}
$\Rightarrow)$ Suppose that there exists a subset of vertices $S \subsetneq V(G)$ verifying Theorem's hypothesis. Now, as for every pair of vertices $v,w \in S$, 
$N_{\overline{S}}(v)=N_{\overline{S}}(w)$, and $|E(G[S])| \geq 1$, we have that every edge in $E(G[S])$ will not be part of any biclique containing an edge outside $E(G[S])$. 
This implies that $KB_e(G)$ is disconnected and proves the ``only if'' part of the Theorem.

$\Leftarrow)$ Suppose that $KB_e(G)$ is not connected. We will show how to find a set $S$ of vertices verifying the hypothesis of the Theorem.
Let $b$ be a vertex in $KB_e(G)$ and let $B$ be its corresponding biclique in $G$.
Let $S_{E(B)} = \{ e \in E(G) : e$ belongs to a biclique $B'$ such that its corresponding vertex $b' \in KB_e(G)$ is in the same connected component as $b\}$.
Clearly $S_{E(B)} \neq E(G)$ and $S_{E(B)} \neq \emptyset$ as $E(B) \subseteq S_{E(B)}$. Let $S_{V(B)} = \{ v \in V(G) : \exists e \in S_{E(B)}$ such that $e$ is incident to $v\}$.
Clearly $S_{V(B)} \neq \emptyset$ as $V(B) \subseteq S_{V(B)}$.
We have the following two cases now.

\begin{itemize}
\item \textbf{Case A)} $E(G[S_{V(B)}]) \subseteq S_{E(B)}$.\\
We will show that $S_{V(B)}$ is the desired set.
Observe that $G[S_{V(B)}]$ has at least one edge and it is clearly connected. Now let $uw$ be an edge such that $u \in S_{V(B)}$ and $w \in \overline{S_{V(B)}}$.
Let $u' \in S_{V(B)}$ be a vertex different to $u$. If $u'$ is adjacent to $w$ there is nothing to show. Suppose then that $u'$ is not adjacent to $w$.
Now, since $G[S_{V(B)}]$ is connected, there is an induced path $u'=u_1u_2\ldots u_k=u$ between $u'$ and $u$. Let $u_i$, $i \in \{2,\ldots,k\}$, be the vertex of minimum index of the path
that is adjacent to $w$. Clearly $u_i$ exists as $u_k=u$ is adjacent to $w$. Since $u_{i-1}$ is not adjacent to $w$, the set 
$\{u_i,u_{i-1},w\}$ is contained 
in a biclique that has the edge $u_{i-1}u_i$. As the edge $u_{i-1}u_i \in S_{E(B)}$, we have that
the edge $u_iw \in S_{E(B)}$ as well, thus
$w \in S_{V(B)}$ which is a contradiction. We conclude that $u'$ should be adjacent to $w$, obtaining that for every pair of vertices $u,u' \in S_{V(B)}$, 
we have that $N_{\overline{S_{V(B)}}}(u)=N_{\overline{S_{V(B)}}}(u')$ as desired.

\item \textbf{Case B)} $\exists e \in E(G[S_{V(B)}]) - S_{E(B)}$.\\
There exists a biclique $B_e$ such that $e \in E(B_e)$ and $B_e$ does not have any edge in common with $S_{E(B)}$.
Consider now the sets $S_{E(B_e)}$ and $S_{V(B_e)}$ defined likewise $S_{E(B)}$ and $S_{V(B)}$. It is clear that $S_{E(B_e)} \neq \emptyset$, $S_{E(B)} \cap S_{E(B_e)} = \emptyset$ and  
$S_{V(B_e)} \neq \emptyset$. Moreover, $S_{V(B_e)} \subseteq S_{V(B)}$. For this, observe that the endpoints of the edge $e$, say $v,w \in B_e$, belong
to $S_{V(B)}$, as $e \in E(G[S_{V(B)}])$. Furthermore, they should have a common neighbor, say $z \in S_{V(B)}$, with $vz, wz \in S_{E(B)}$. 
If $B_e$ is just the edge $e$, we obtain directly that $S_{V(B_e)} \subseteq S_{V(B)}$. If $B_e$ is a larger biclique, there exists a vertex $u \in B_e$, without loss of generality, adjacent to $v$ and not adjacent to $w$.
Clearly, $u$ belongs to $S_{V(B_e)}$. Now if $u$ is not adjacent to $z$, then there is a biclique containing
$\{u,v,z\}$ and since $vz \in S_{E(B)}$, then $uv \in S_{E(B)}$ and therefore, $e=vw \in S_{E(B)}$ which is a contradiction.
Finally, $u$ is adjacent to $z$, thus $\{u,z,w\}$ is contained in a biclique containing the edge $zw \in S_{E(B)}$, therefore
$uz \in S_{E(B)}$ implying that $u \in S_{V(B)}$. Since $G[S_{V(B_e)}]$ is connected, similar arguments can be applied 
for every other vertex $x \in S_{V(B_e)}$. As there is an induced path from $x$ to vertices $u,w$, we will obtain that $xz \in S_{E(B)}$ and 
thus $x \in S_{V(B)}$, otherwise we would get a contradiction the same way as before.

Now, if $|S_{V(B_e)}| = 2$ then it is easy to see that $S_{V(B_e)}$ is the desired set. In what follows we assume that 
$|S_{V(B_e)}| \geq 3$.

We will show that $S_{V(B_e)} \subsetneq S_{V(B)}$ therefore, if $S_{V(B_e)}$ does not verify \textbf{Case A)}, then we obtain another set $S_{V(B_{e'})}$ such that
$S_{V(B_{e'})} \subsetneq S_{V(B_e)}$ and we repeat the process. Since the graph is finite, in some point we will obtain a set of vertices verifying \textbf{Case A)} which will
conclude the proof (see Fig.~\ref{congraph}). We will use the following three claims.
\begin{claim}
$\forall e=vw \in S_{E(B_e)}$ that belongs to two different bicliques, then 
\begin{itemize}
\item $\exists v_1,w_1 \in S_{V(B_e)}$ such that the pairs of vertices $v,v_1$ and $w,w_1$ are adjacent, and $v,w_1$, $w,v_1$ and $v_1,w_1$ are not adjacent; or
\item $\exists v_1,v_2 \in S_{V(B_e)}$ such that the pairs of vertices $v,v_1$, $v,v_2$ and $v_1,v_2$ are adjacent, and $w,v_1$ and $w,v_2$ are not adjacent.\footnote{Note that this Claim is valid for any edge in a graph that belongs to two bicliques.}
\end{itemize}
These two options are shown in Figure~\ref{casosclaim}.
\end{claim}
\begin{figure}[ht!]	
  \centering	
  \includegraphics[scale=.5]{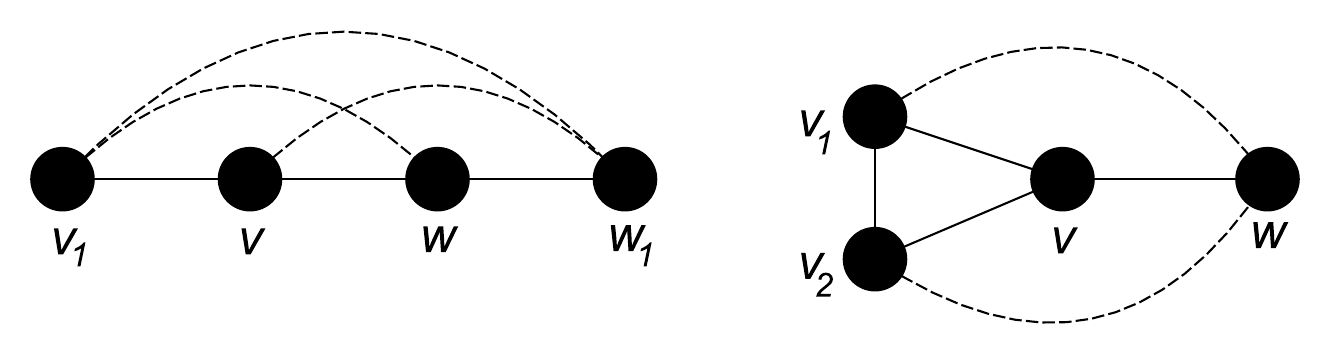} 	
  \caption{Unique two options for an edge $vw$ belonging to two different bicliques.}
  \label{casosclaim}	
\end{figure}
\vspace*{-3mm}
\textit{Proof of Claim $1$}. First observe that since $|S_{V(B_e)}| \geq 3$, there exists a vertex, say $v_1 \in S_{V(B_e)}$, adjacent to $v$ and not adjacent to $w$, i.e., the biclique
containing the edge $vw$ is bigger than a $K_{1,1}$. This implies that $vv_1 \in S_{E(B_e)}$.  Now, since $vw$ belongs to another biclique than the one containing $\{v,w,v_1\}$, then, one case 
would be to have a vertex, say $w_1 \in S_{V(B_e)}$, adjacent to $w$ and not adjacent to $v$. Moreover, $w_1$ is not adjacent to $v_1$, as otherwise, $\{v,w,v_1,w_1\}$ would be in the same biclique.
Clearly, the edge $ww_1 \in S_{E(B_e)}$, as both bicliques intersect in the edge $vw \in S_{E(B_e)}$. This shows the first option of the Claim. Now,
if such a vertex $w_1$ does not exist, then it should exist a vertex $v_2 \in S_{V(B_e)}$ such that $v_2$ is adjacent to $v$ and not adjacent to $w$. Moreover, since the biclique
containing $\{v,w,v_2\}$ should be different to the one having $\{v,w,v_1\}$, this vertex $v_2$ is adjacent to $v_1$. As before, since these two bicliques have $vw \in S_{E(B_e)}$ in common,
then $vv_2 \in S_{E(B_e)}$ as well. Note that the edge $v_1v_2$ might or might not belong to $S_{E(B_e)}$.
\qqed

\begin{claim}
Let $e=u_1u_2 \in E(G[S_{V(B)}]) - S_{E(B)}$ and let $x \in S_{V(B)}$ such that $u_1,u_2 \in N(x)$. Then $\forall v \in S_{V(B_e)}$, $v$ is adjacent to $x$ and the
edge $vx \in S_{E(B)}$. 
\end{claim}
\vspace*{-3mm}
\textit{Proof of Claim $2$}. First note that $u_1x$ and $u_2x$ belong to $S_{E(B)}$ because $u_1u_2 \in E(G[S_{V(B)}])$. Now, since $|S_{V(B_e)}| \geq 3$, the biclique $B_e$ containing 
$u_1u_2$ is bigger than a $K_{1,1}$. Let $u_3 \in V(B_e) \subseteq S_{V(B_e)}$ be a vertex different from $u_1$ and $u_2$ such that (without loss of generality) $u_3,u_1$ are not adjacent and
$u_3,u_2$ are adjacent. Clearly, the edge $u_3u_2 \in S_{E(B_e)}$. If $u_3$ and $x$ are not adjacent, then 
$u_2x \in S_{E(B)} \cap S_{E(B_e)}$ (as $\{u_2,x,u_3\}$ is contained in a biclique that intersects $B_e$),
which is a contradiction. Therefore, $u_3,x$ are adjacent. Now, since the set $\{x,u_3,u_1\}$ is contained in a biclique that has the edge $u_1x \in S_{E(B)}$, 
it follows that $u_3x \in S_{E(B)}$ as well. This same argument applies for every vertex in $B_e$, therefore if $V(B_e) = S_{V(B_e)}$, then the proof is complete. Otherwise, there exists another
biclique $B'$ in $S_{V(B_e)}$ having an edge in common with $B_e$. Suppose without loss of generality that the edge $u_1u_2$ belongs to both. Now, by Claim 1, there are two options for 
this situation. Observe that in both options, there exists a vertex, say $u_4 \in S_{V(B_e)}$, that is adjacent to $u_1$ and not to $u_2$, or adjacent to $u_2$ and not $u_1$. That is, there is 
an induced $P_3$ containing $u_4$ in one of the extremes. Suppose the first case, i.e., $u_4$ is adjacent to $u_1$ and not to $u_2$ (the other option is similar).
As before, $u_4$ must be adjacent to $x$, otherwise $u_1x \in S_{E(B)} \cap S_{E(B_e)}$, a contradiction. Since $\{u_4,x,u_2\}$ is contained in a biclique that has the edge 
$u_2x \in S_{E(B)}$, then we have that $u_4x \in S_{E(B)}$. Observe that this argument can be used for every vertex in $B'$. Finally, we apply the same reasoning we used
for $B'$, to the other bicliques having edges in $S_{E(B_e)}$ that intersect previous analyzed bicliques. This completes the proof.
\qqed

\begin{claim}
There exists a vertex $x \in S_{V(B)}$ such that $x \notin S_{V(B_e)}$.
\end{claim}
\vspace*{-3mm}
\textit{Proof of Claim $3$}. By Claim 2, we have a vertex $x \in S_{V(B)}$ such that $\forall v \in S_{V(B_e)}$, $v$ is adjacent to $x$ and the edge $vx \in S_{E(B)}$. 
Finally, since $S_{E(B)} \cap S_{E(B_e)} = \emptyset$, and by definition of the sets $S_{E(B_e)}$ and $S_{V(B_e)}$, we have that $x \notin S_{V(B_e)}$.
\qqed

To conclude the proof of \textbf{Case B)}, by Claim 3, there is a vertex $x \in S_{V(B)}$ such that $x \notin S_{V(B_e)}$, thus $S_{V(B_e)} \subsetneq S_{V(B)}$ as we wanted to show.
Therefore, we can always obtain a set of vertices veryfing \textbf{Case A)} as desired.
\end{itemize}
As there are no cases left to analyze, the proof is complete.
\end{proof}

\begin{figure}[ht!]	
  \centering	
  \includegraphics[scale=.3]{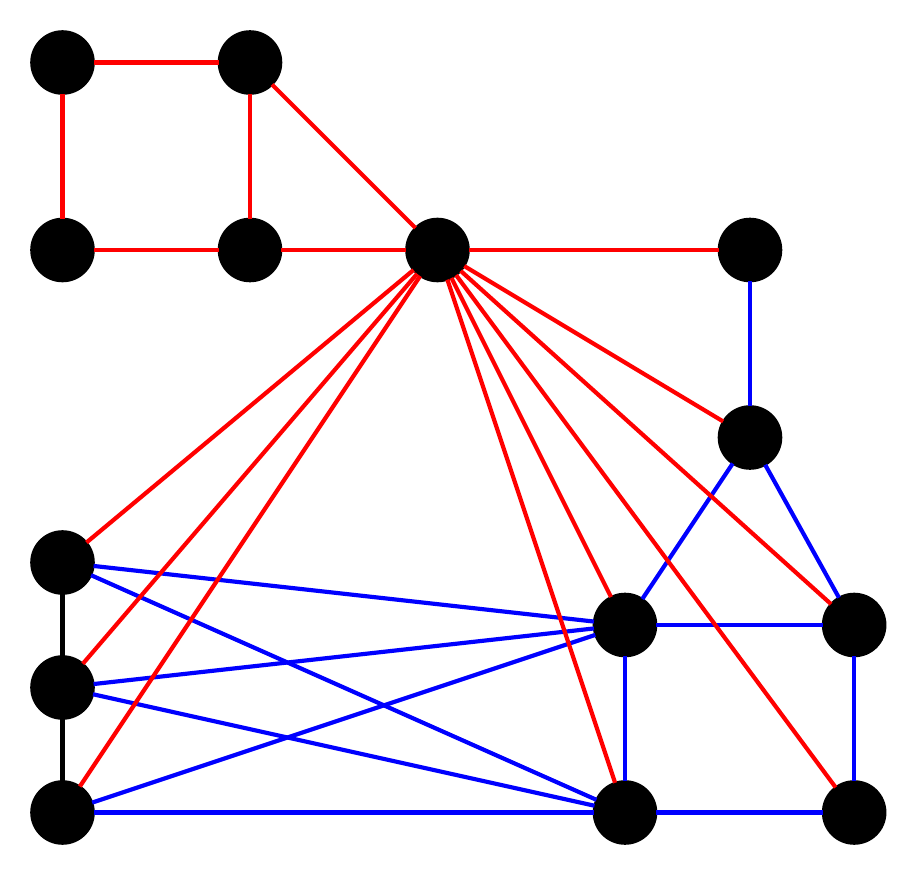} 	
  \caption{In this example we can see three set of edges, $S_{E(red)}$, $S_{E(blue)}$ and $S_{E(black)}$. Note that $S_{V(black)} \subsetneq S_{V(blue)} \subsetneq S_{V(red)} = V(G)$, then
  following \textbf{Case B)} of the proof of Theorem~\ref{conexo}, the set $S_{V(black)}$ is the desired one.}
  \label{congraph}	
\end{figure}
\FloatBarrier

To finish the section, we present an $O(n \times m)$ algorithm that, given a graph $G$, decides if $KB_e(G)$ is connected or not. Moreover, if $KB_e(G)$ is disconnected, the algorithm gives 
a partition of the edges of $G$ such that each set of the partition has the edges belonging to bicliques that are in the same connected component in $KB_e(G)$. This algorithm
relies mostly in Claim 1 of Theorem~\ref{conexo}, since otherwise, verifying the condition for all subsets of vertices $S \subsetneq V(G)$ would take exponential time.
We also remark that, since the number of bicliques of a graph can be exponential~\cite{PrisnerC2000}, constructing $KB_e(G)$ to check later if it is connected can take exponential time as well.

\begin{algorithm}[H]
\SetAlgoLined
\SetKwInOut{Input}{Input}
\SetKwInOut{Output}{Output}
\Input{A connected graph $G$.}
\Output{A partition of $E(G)=E_1 \cup \cdots \cup E_k$ such that each $E_i$, for $i=1,\ldots,k$, has the edges belonging to bicliques that are in the same connected component in $KB_e(G)$.}
\BlankLine

 mark all edges as not used; $k \leftarrow 0$; $S_E \leftarrow \emptyset$\;
 \While{there exist unused edges}{
	$k \leftarrow k+1$\; 
	take an unused edge $e$; $S_E \leftarrow S_E \cup \{e\}$; mark $e$ as used\;
  	\While{$S_E \neq \emptyset$}{
  		remove an edge $e=vw \in S_E$\;
  		$E_k \leftarrow E_k \cup \{e\}$\;
  		\For{every vertex $z \in N(v)-N(w)$ and $zv$ not used}{
 			$S_E \leftarrow S_E \cup \{zv\}$; mark $zv$ as used\;
 		}
  		\For{every vertex $z \in N(w)-N(v)$ and $zw$ not used}{
 			$S_E \leftarrow S_E \cup \{zw\}$; mark $zw$ as used\;
 		}
  	}
 }
 \eIf{$k=1$}{
 	\Return{$KB_e(G)$ is connected}\;
 }{
 	\Return{$KB_e(G)$ is disconnected and $E(G)=E_1 \cup \cdots \cup E_k$}\; 
 }
 \caption{Connectivity of $KB_e(G)$}
 \label{algo1}
\end{algorithm}
\vspace*{2mm}
It is clear that Algorithm~\ref{algo1} runs in $O(n \times m)$ since each edge is added once to $S_E$ and each time we check all its endpoint's neighbors.
It only remains to show that the algorithm is correct.

\begin{proposition}
Algorithm~\ref{algo1} correctly finds a partition of $E(G)=E_1 \cup \cdots \cup E_k$ such that each $E_i$, for $i=1,\ldots,k$, has the edges belonging to bicliques that are in the same 
connected component in $KB_e(G)$. In particular, $KB_e(G)$ is connected if and only if $k=1$, that is, $E_1 = E(G)$.
\end{proposition}
\begin{proof}
Observe that in the first while loop, the algorithm takes an unused edge (while exists) and adds it to a set $S_E$ of edges to analyze. 
The second while loop will add all edges that belong to the partition of that edge. For this, it takes an edge $e=vw$ from $S_E$ (while $S_E \neq \emptyset$) and
adds it to the current edge partition. Now, if $e$ is not a biclique itself (in which case $N(v)=N(w)$, thus $e$ is alone in its partition), then it must
exist other vertex $z$ verifying $z \in N(v)-N(w)$ or $z \in N(w)-N(v)$, therefore it adds all edges of the form $zv$ or $zw$ to $S_E$ and to the current partition, respectively.
Now, for each other iteration of the second while loop, the algorithm uses Claim 1 of Theorem~\ref{conexo} to see if an already used edge belongs to another biclique, and adds
these new edges corresponding to those bicliques. When the while loop ends for an iteration $i$, that is, $S_E = \emptyset$, then $E_i$ has all edges that belong to bicliques
that are in the same connected component of $KB_e(G)$ as the biclique containing the initial edge of that iteration.

Finally, if $E_1 = E(G)$, then $KB_e(G)$ is connected since all edges of the graph belong to bicliques to one same connected component in $KB_e(G)$. Otherwise,
one of the sets $S_{V_i}$ formed with incident vertices to the edges in $E_i$ (analog definition as $S_{E(B)}$ and $S_{V(B)}$ in Theorem~\ref{conexo}) verifies
that $S_{V_i} \subsetneq V(G)$ and therefore Theorem~\ref{conexo} holds, that is, $KB_e(G)$ is disconnected.
\end{proof}

\section{Convergence}

To start this section we have this first easy result.

\begin{lemma}
For $n\geq 2$, the complete graph $K_n$ converges to the empty graph under the operator $KB_e$ in two steps.
\end{lemma}
\begin{proof}
Clearly each edge of $K_n$ is a biclique that does not edge-intersect with another one. Then $KB_e(K_n)$ consists of $\frac{n(n-1)}{2}$ isolated vertices (and no bicliques), 
therefore $KB^2_e(K_n)$ is the empty graph.
\end{proof}

Next we show that graphs without induced cycles of length $3$ and $4$ are convergent.

\begin{theorem}\label{conv1}
If $G$ has girth at least five, then the edge-biclique operator applied to $G$ converges towards the graph induced by the union of all the cycles and paths connecting cycles of $G$. 
\end{theorem}
\begin{proof}
If $G$ has girth at least five, then every biclique is a star. Moreover $G$ has no triangles, so $N(v)$ is a stable set and thus, for each $v$ of degree more than one, $N[v]$ is a maximal biclique. 
Notice also that if $u$ is adjacent to $v$, $N[u]$ and $N[v]$ contain a common edge, therefore the vertices in $KB_e(G)$ corresponding to the bicliques $N[u]$ and $N[v]$ will be adjacent. 
We can conclude that $KB_e(G)$ is exactly the graph induced by all vertices of degree at least two of $G$. For $k$ big enough, the only vertices left in $KB_e^k(G)$ are those which belong to cycles
or to paths connecting cycles, that is, $G$ converges under the operator $KB_e$ towards the graph induced by the cycles and paths connecting cycles of $G$.
\end{proof}

As an immediate result of Theorem~\ref{conv1}, we obtain the following corollary.

\begin{corollary}\label{conv2}
If $G$ has girth at least five and has no vertices of degree one, then $KB_e(G) = G$.
\end{corollary}

One natural question that arises from Corollary~\ref{conv2} is: Given a graph $G$ such that $KB_e(G)=G$, does $G$ have girth at least five and no vertices of degree one?
The answer is no, for instance, the graph $\overline{C_7}$ shown in Figure~\ref{grafconv} satisfies that $KB_e(G) = G$ but its girth is three\footnote{Found using the computer.}.

\begin{figure}[ht!]	
  \centering	
  \includegraphics[scale=.3]{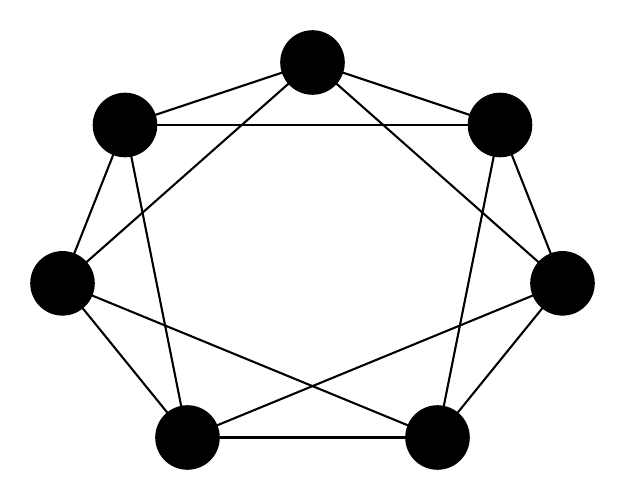} 	
  \caption{The graph $\overline{C_7}$ is the smallest graph satisfying $KB_e(G) = G$ with girth less than five.}
  \label{grafconv}	
\end{figure}

From Theorem~\ref{conv1}, we also obtain the following results.

\begin{corollary}
For every $k\geq 1$, there is a graph that converges in $k$ steps under the operator $KB_e$.
\end{corollary}
\begin{proof}
Just take any induced cycle $C_n$, $n \geq 5$, and join one of its vertices to the endpoint of a path $P_k$. Observe that this graph converges to $C_n$ in exactly $k$ steps 
(see Fig~\ref{kpasos}).
\end{proof}

\begin{figure}[ht!]	
  \centering	
  \includegraphics[scale=.25]{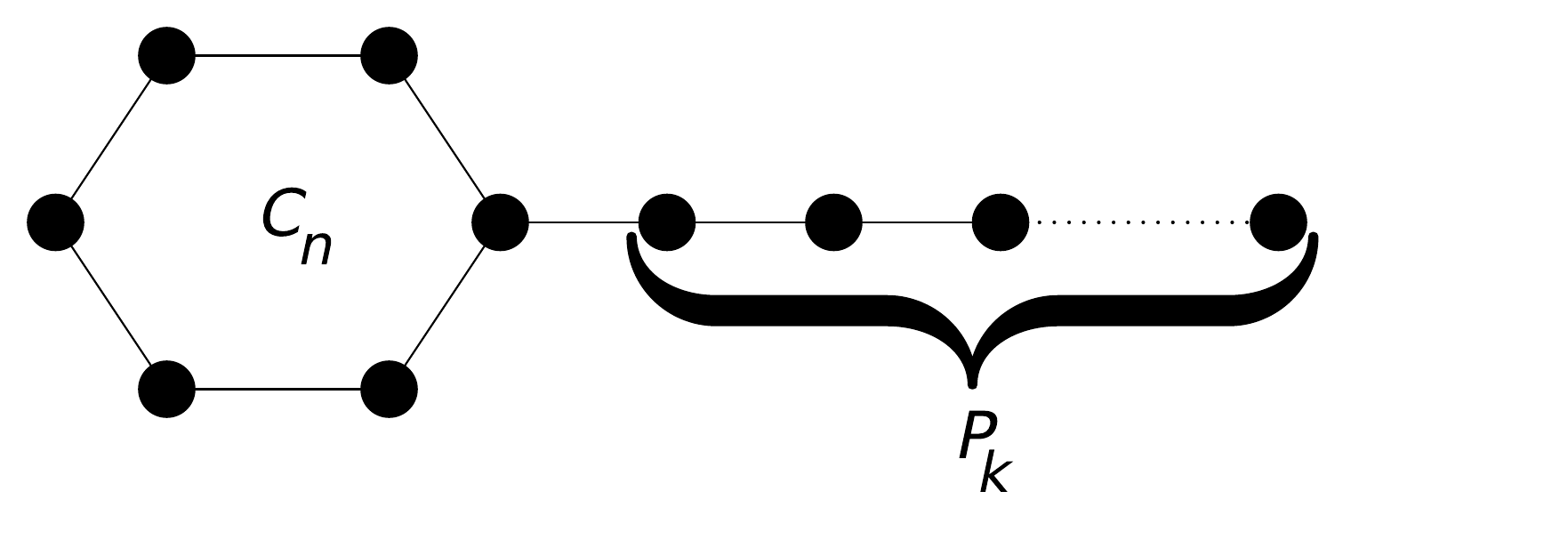} 	
  \caption{Graph $G$ that converges in $k$ steps under the operator $KB_e$}
  \label{kpasos}	
\end{figure}

\begin{corollary}
Trees converge to the empty graph under the operator $KB_e$.
\end{corollary}

\section{Divergence}

In this section we study the divergence of the operator $KB_e$. We start with the following definition.

\begin{definition}
Let $G$ be a graph and let $C=v_0v_1\ldots v_{n-1}$ be an induced cycle of length $n \geq 5$. We say that $C$ has \textit{good neighbors} whenever for all vertices $v \in V(G)-C$, if $\{v_{i-1},v_{i+1}\} \subseteq N(v)$ then $v_i \in N(v)$, for $i=0,\ldots,n-1$ and all subindices taken $(mod \ n)$. (see Fig~\ref{goodvecinos}).
\end{definition}

\begin{figure}[ht!]	
  \centering	
  \includegraphics[scale=.3]{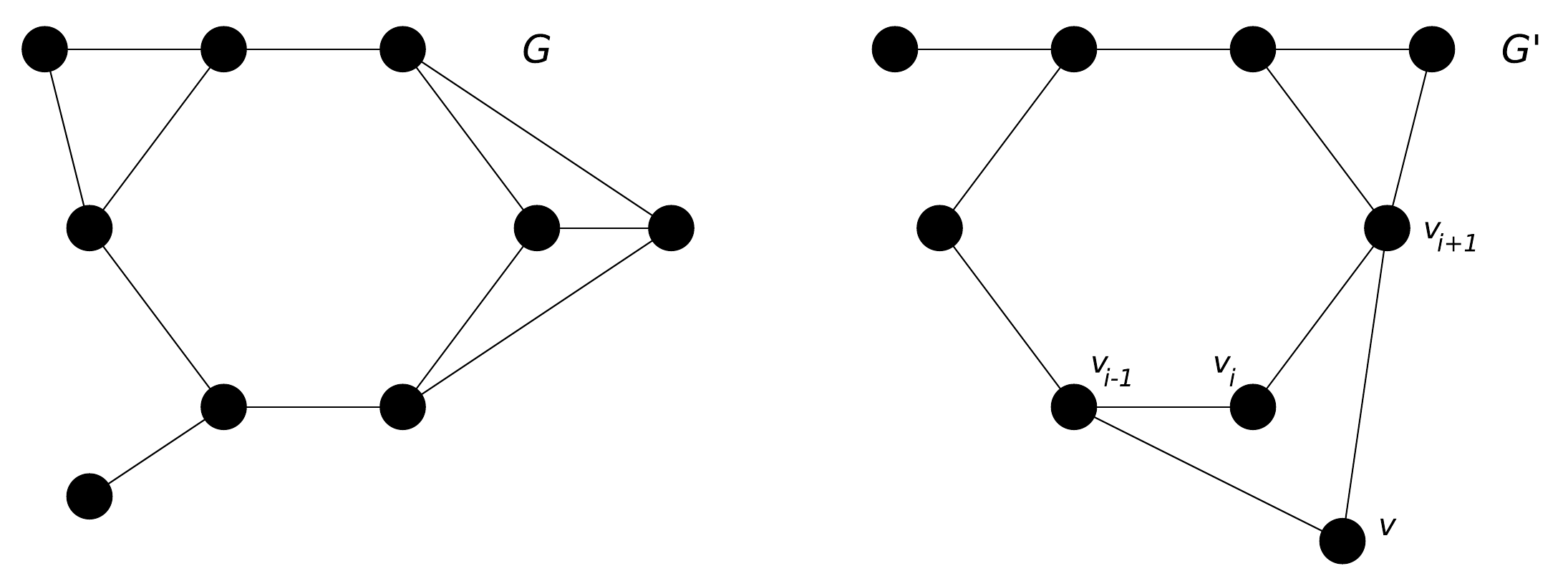} 	
  \caption{$G$ has a cycle with \textit{good neighbors} while $G'$ has not, since $v$ is adjacent to $v_{i-1}$ and $v_{i+1}$ but not adjacent to $v_i$.}
  \label{goodvecinos}	
\end{figure}

Now we present an important proposition that assures that the good neighbors property is invariant through the iterations of the operator $KB_e$.

\begin{proposition}\label{condsylvain}
Let $G$ be a graph and let $C=v_0v_1\ldots v_{n-1}$ be an induced cycle of length $n \geq 5$ with good neighbors. 
Let $B_i$, $i=0,\ldots,n-1$, be bicliques in $G$ containing the vertices $\{v_{i-1},v_i,v_{i+1}\}$ $(mod \ n)$, respectively, and let 
$b_i$, $i=0,\ldots,n-1$, be the vertices in $KB_e(G)$ corresponding to the bicliques $B_i \in G$. Then $V(B_i) \subseteq N[v_i]$ and $C'= b_0b_1\ldots b_{n-1}$ is an induced cycle of $KB_e(G)$.
Moreover, $C'$ has good neighbors.
\end{proposition}
\begin{proof}
As $C$ is an induced cycle in $G$, let $B_i$, $i=0,\ldots,n-1$, be bicliques that contain the vertices $\{v_{i-1},v_i,v_{i+1}\}$ $(mod \ n)$, respectively. Clearly, each $B_i$ intersects $B_{i+1}$
in the edge $v_iv_{i+1}$, therefore if we call $b_i$, $i=0,\ldots,n-1$, the corresponding vertices in $KB_e(G)$ to the bicliques $B_i$, then we have that $b_0b_1\ldots b_{n-1}$ form a cycle $C'$ in $KB_e(G)$.
Now, let $v \in G$ be a vertex in $B_i - \{v_{i-1}, v_i, v_{i+1}\}$. As $B_i$ is a biclique of $G$, either $v$ is adjacent to $v_{i-1}$ and $v_{i+1}$ but not adjacent to $v_i$, which is not possible because 
$C$ has good neighbors, or $v$ is adjacent to $v_i$. Therefore, for all $i =0,\ldots, n-1$, $V(B_i) \subseteq N[v_i]$ and $C'$ is an induced cycle of $KB_e(G)$.

Now, let $b \in V(KB_e(G))-C'$ be a vertex such that $\{b_{i-1}, b_{i+1}\} \subseteq N(b)$ for some $i$. If $B$ is the biclique of $G$ corresponding to the vertex $b \in KB_e(G)$, then $B$ contains 
$v_{i-1}$ and $v_{i+1}$, since $V(B_{i-1}) \subseteq N[v_{i-1}]$ and $V(B_{i+1}) \subseteq N[v_{i+1}]$. As $v_{i-1}$ and $v_{i+1}$ are not adjacent in $G$, there exists a vertex 
$v\in B \cap B_{i-1} \cap B_{i+1}$ such that $v$ is adjacent to both $v_{i-1}$ and $v_{i+1}$. If $v\neq v_i$, since $C$ has good neighbors, $v$ must also be adjacent to $v_i$, contradicting the fact 
that $v \in B_{i-1}$ (or $B_{i+1}$). Therefore, $v = v_i$ and $B$ and $B_i$ have an edge in common, that is, $b$ is adjacent to $b_i$ in $KB_e(G)$ and thus $C'$ has good neighbors. 
\end{proof}

Before the main theorem, we define the following family of graphs.
\begin{definition}
For $n \geq 3$ and $m \geq 1$, the $(n,m)-necklace$ graph on $n+m$ vertices consists of an induced cycle $C_n$ and a complete graph $K_m$, such that for an edge $e \in C_n$, every vertex of the 
$K_m$ is only adjacent to both endpoints of $e$ (see Fig~\ref{necklace}).
\end{definition}

\begin{figure}[ht!]	
  \centering	
  \includegraphics[scale=.3]{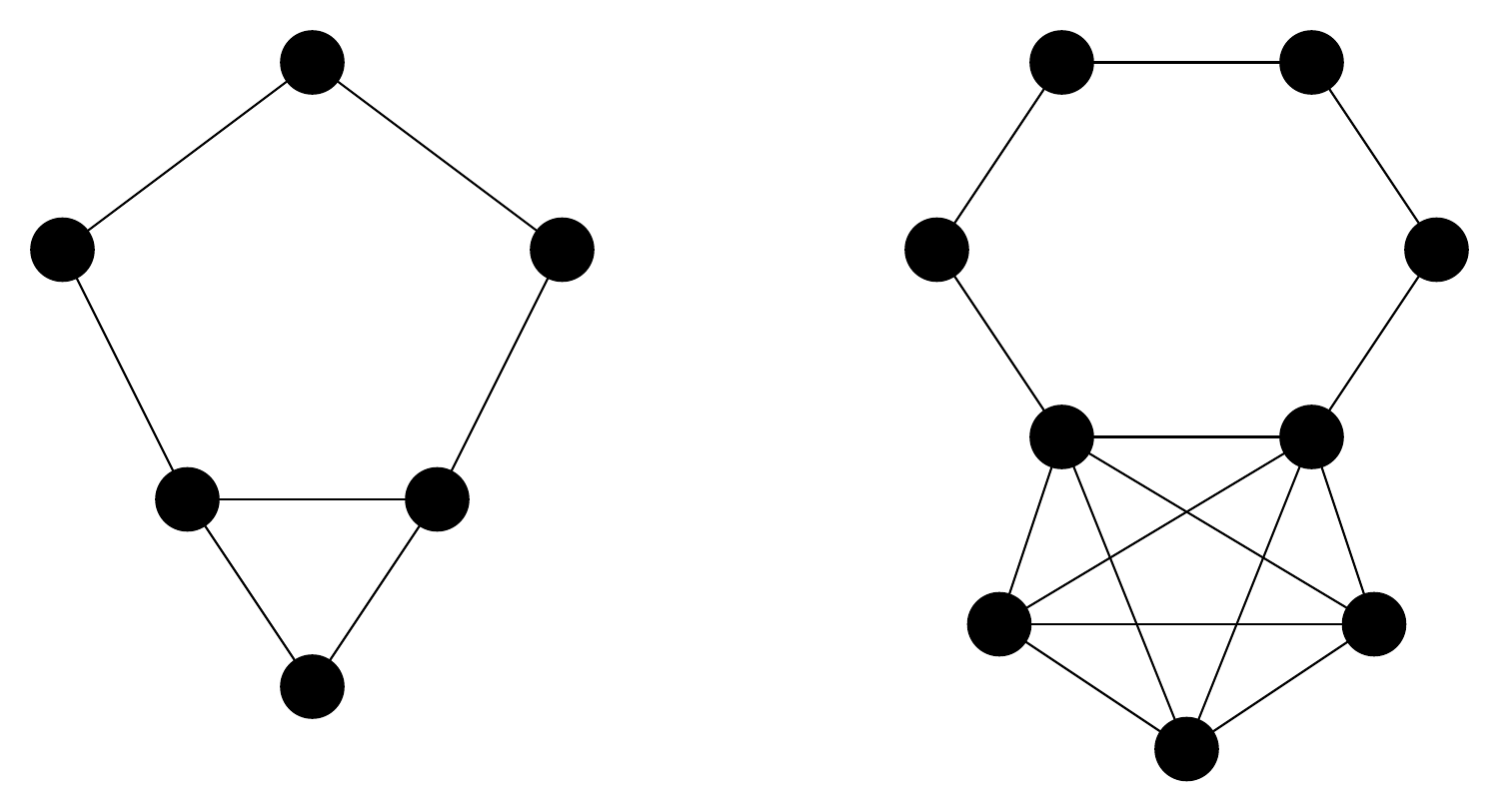} 	
  \caption{$(5,1)-necklace$ and $(6,3)-necklace$ graphs.}
  \label{necklace}	
\end{figure}

Now we present the main theorem of this section.

\begin{theorem}\label{neck}
Let $G$ be a graph that contains an induced $(n,m)-necklace$, $n \geq 5$, $m \geq 1$, such that its cycle has good neighbors. Then, either $KB^2_e(G)$ or $KB^3_e(G)$ contains an induced $(n,m')-necklace$ such that its cycle has good neighbors, and $m' > m$.
\end{theorem}
\begin{proof}
Let $C_n=v_0v_1\ldots v_{n-1}$ be the induced cycle and $K_m= \{w_1,\ldots,w_m\}$ be the complete graph of the $(n,m)-necklace$, respectively. Let $v_iv_{i+1}$, for some 
$i \in \{0,\ldots,n-1\}$ $(mod \ n)$, be the edge of the $C_n$ such that
$w_j$ is adjacent to $v_i$ and $v_{i+1}$ for all $j=1,\ldots,m$. Let $B_t$, $t=0,\ldots,n-1$, be bicliques that contain the vertices $\{v_{t-1},v_t,v_{t+1}\}$ $(mod \ n)$, respectively, and let 
$b_t$, $t=0,\ldots,n-1$, be the corresponding vertices in $KB_e(G)$ to the bicliques $B_t$. By Proposition~\ref{condsylvain}, $C'_n = b_0b_1\ldots b_{n-1}$ is an induced cycle in $KB_e(G)$ with good neighbors.

Consider these two families of bicliques $B^1 = \{B^1_j : \{w_j,v_i,v_{i-1}\} \subseteq B^1_j, j=1,\ldots,m\}$ and 
$B^2 = \{B^2_j : \{w_j,v_{i+1},v_{i+2}\} \subseteq B^2_j, j=1,\ldots,m\}$. Clearly, all these $2m$ bicliques are different and moreover, they are different to the bicliques $B_t$ for $t=0,\ldots,n-1$ as $C_n$ has good neighbors. Now we can see that $(\bigcap_{j=1}^m B^1_j) \cap B_{i-1} \cap B_i = \{v_{i-1},v_i\}$ and 
$(\bigcap_{j=1}^m B^2_j) \cap B_{i+1} \cap B_{i+2} = \{v_{i+1},v_{i+2}\}$.
Therefore if $b^1_j$ and $b^2_j$, $j=1,\ldots,m$, are the corresponding vertices in $KB_e(G)$ to the bicliques $B^1_j$ and $B^2_j$, we have
that in $KB_e(G)$, $K^1_m = \{b^1_1,\ldots,b^1_m\}$ and $K^2_m = \{b^2_1,\ldots,b^2_m\}$ are two complete graphs such that $b^1_j$ is adjacent to $b_{i-1}$ and $b_i$, and $b^2_j$ is adjacent to $b_{i+1}$ and $b_{i+2}$, for all $j=1,\ldots,m$. Notice that as $C_n$ has good neighbors in $G$, then in $KB_e(G)$ we have $N(b^1_j) \cap C'_n = \{b_{i-1},b_i\}$ and 
$N(b^2_j) \cap C'_n = \{b_{i+1},b_{i+2}\}$, for all $j=1,\ldots,m$ (see Fig~\ref{divergeproof2}).

\begin{figure}[ht!]	
  \centering	
  \includegraphics[scale=.4]{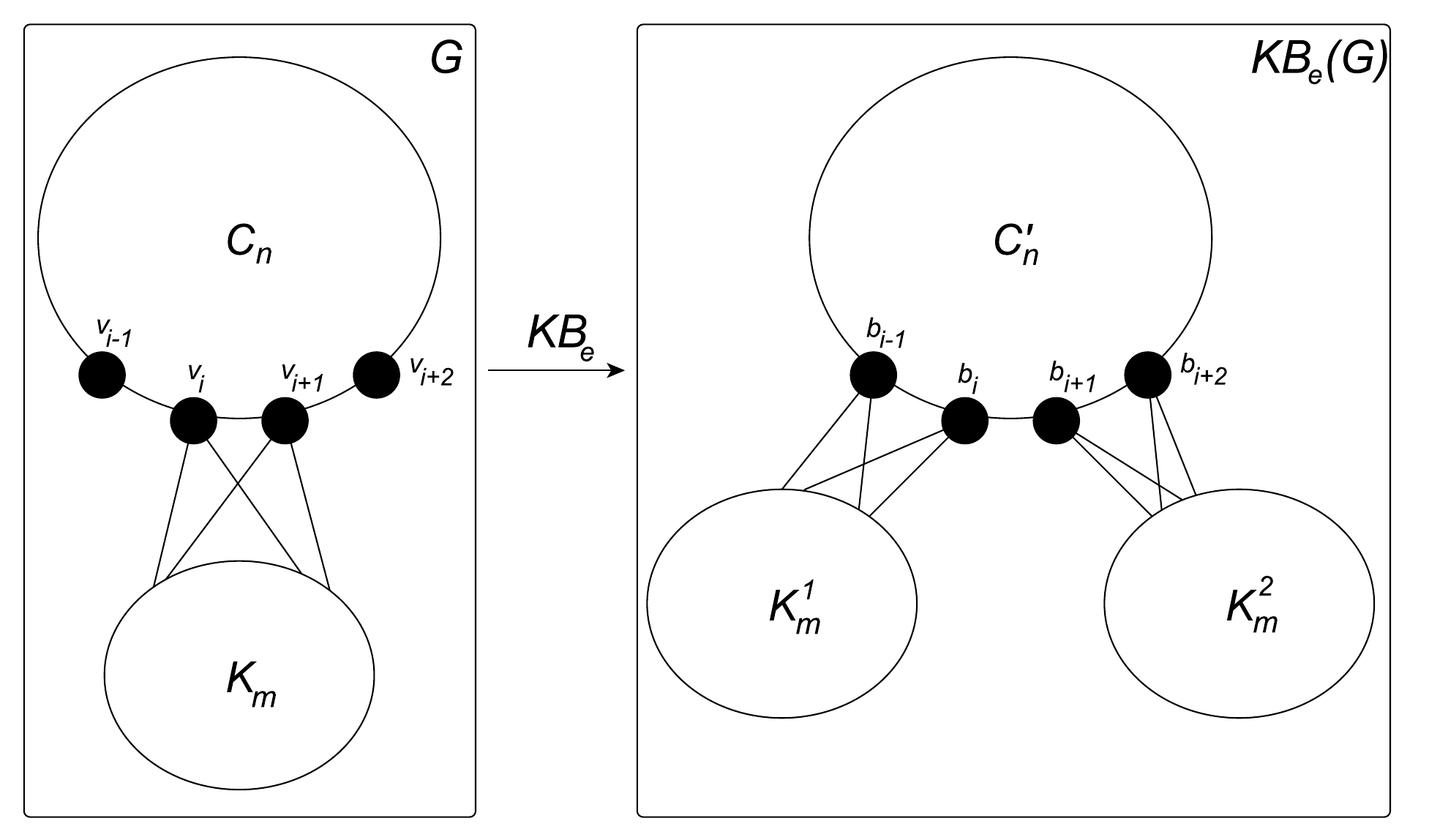} 	
  \caption{First iteration of the operator $KB_e$ applied to $G$ containing an induced $(n,m)-necklace$ with good neighbors.}
  \label{divergeproof2}	
\end{figure}

Now, let $\widetilde{B}_t$, $t=0,\ldots,n-1$, be the bicliques of $KB_e(G)$ that contain the vertices $\{b_{t-1},b_t,b_{t+1}\}$ $(mod \ n)$, respectively, and
$\widetilde{b}_t$, $t=0,\ldots,n-1$, be the corresponding vertices in $KB^2_e(G)$ to the bicliques $\widetilde{B}_t$. Again, by Proposition~\ref{condsylvain}, 
$C''_n = \widetilde{b}_0\widetilde{b}_1\ldots \widetilde{b}_{n-1}$ is an induced cycle in $KB^2_e(G)$ with good neighbors.

Now for each $b^1_j$, $j=1,\ldots,m$, we have that $\{b^1_j,b_i,b_{i+1}\}$ is contained in a biclique $\widetilde{B}^1_j$. Similarly, 
for each $b^2_j$, $j=1,\ldots,m$, $\{b^2_j,b_i,b_{i+1}\}$ is contained in a biclique $\widetilde{B}^2_j$. 
In the worst case (to minimize the number of bicliques), if there is exactly a perfect matching between $K^1_m$ and $K^2_m$, say $b^1_j$ is adjacent to $b^2_j$, for each $j=1,\ldots,m$,
then $\widetilde{B}^1_j= \widetilde{B}^2_j$. 
We have the following two cases:

\textbf{Case A}: \textit{There is at least one vertex $b^1_1 \in K^1_m$ not adjacent to any vertex of $K^2_m$.}
Clearly, the biclique $\widetilde{B}^1_1$ is different to the $m$ bicliques $\widetilde{B}^2_j$, for all $j=1,\ldots,m$, and furthermore, these $m+1$ bicliques are different to the bicliques 
$\widetilde{B}_t$ for $t=0,\ldots,n-1$. Observe that $(\bigcap_{j=1}^m \widetilde{B}^2_j) \cap {B}^1_1 = \{b_i,b_{i+1}\}$ and moreover, 
$\widetilde{B}_i \cap \widetilde{B}_{i+1} = \{b_i,b_{i+1}\}$. Therefore, if $\widetilde{b}^1_1$ and $\widetilde{b}^2_j$, $j=1,\ldots,m$, are the corresponding vertices in $KB^2_e(G)$ to the 
bicliques $\widetilde{B}^1_1$ and $\widetilde{B}^2_j$, respectively, we have that in $KB^2_e(G)$, $\{\widetilde{b}^1_1,\widetilde{b}^2_1,\ldots,\widetilde{b}^2_m\}$ is a complete graph on $m+1$ vertices such that, as $C'_n$ has good neighbors, every vertex of this $K_{m+1}$ is only adjacent to $\widetilde{b}_i$ and to $\widetilde{b}_{i+1}$ on the cycle $C''_n$. That is, $KB^2_e(G)$ 
contains an induced $(n,m+1)-necklace$ such that its cycle $C''_n$ has good neighbors. See Fig~\ref{divergeproof4}.

\begin{figure}[ht!]	
  \centering	
  \includegraphics[scale=.4]{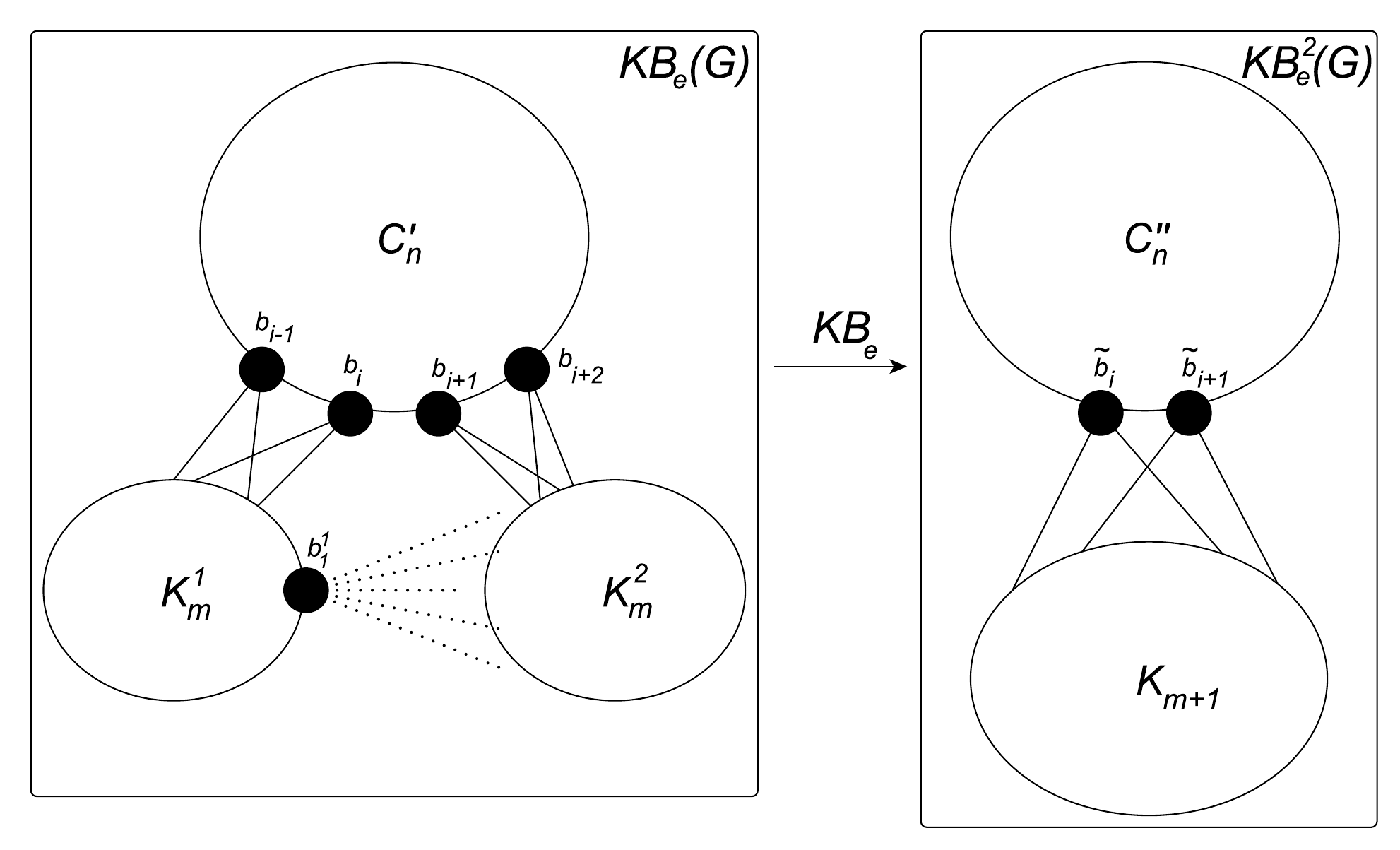} 	
  \caption{\textbf{Case A}: Second iteration of the operator $KB_e$.}
  \label{divergeproof4}	
\end{figure}

\textbf{Case B}: \textit{Every vertex of $K^1_m$ is adjacent to at least one vertex of $K^2_m$ (and by symmetry every vertex of $K^2_m$ is adjacent to at least one vertex of $K^1_m$).}
As explained above, the worst case is when there is a perfect matching between $K^1_m$ and $K^2_m$. Without loss of generality, suppose that $b^1_j$ is adjacent to $b^2_j$ for each $j=1,\ldots,m$, otherwise we would obtain at least $m+1$ bicliques having the edge $b_ib_{i+1}$ in common and therefore $KB^2_e(G)$ will contain an induced $(n,m+1)-necklace$ such that its cycle $C''_n$ has good neighbors. 
As there is a matching between $K^1_m$ and $K^2_m$, let $\widetilde{B}'_j$ be the bicliques that contain the set $\{b^1_j,b_i,b_{i+1},b^2_j\}$ for each $j=1,\ldots,m$. These bicliques contain
the edge $b_ib_{i+1}$ and they are different to the bicliques $\widetilde{B}_t$ for $t=0,\ldots,n-1$.
Then, if $\widetilde{b}'_j$, $j=1,\ldots,m$, are the corresponding vertices in $KB^2_e(G)$ to the bicliques $\widetilde{B}'_j$, we have that in $KB^2_e(G)$, 
$\{\widetilde{b}'_1,\ldots,\widetilde{b}'_m\}$ is a complete graph on $m$ vertices such that, as $C'_n$ has good neighbors, every vertex of this $K_m$ is only adjacent to $\widetilde{b}_i$ and to 
$\widetilde{b}_{i+1}$ on the cycle $C''_n$. 

Now for each $b^1_j$, $j=1,\ldots,m$, we have that $\{b^1_j,b_{i-1},b_{i-2}\}$ is contained in a biclique $\widetilde{B}^1_j$. All these $m$ bicliques have the edge $b_{i-1}b_{i-2}$ in common. In addition, they
are clearly different to the bicliques $\widetilde{B}_t$, $t=0,\ldots,n-1$ and $\widetilde{B}'_j$, $j=1,\ldots,m$. Suppose now that there is an edge in common between the bicliques, say $\widetilde{B}^1_1$
and $\widetilde{B}'_1$. Then, there must exist a vertex $b \in KB_e(G)$ adjacent to $b_{i-2},b^1_1$ and $b_{i+1}$. This implies that in $G$, there must exist a biclique $B$ 
(corresponding to the vertex $b \in KB_e(G)$) that has edges in common with the bicliques $B_{i-2},B^1_1$ and $B_{i+1}$. Therefore, as $C_n$ has good neighbors, there must exist a vertex $v \in B$ adjacent to
the vertices $v_{i-2}$ and $v_{i+1}$. Finally, as $B$ has an edge in common with the biclique $B^1_1$, $v$ must be adjacent to either to $v_i$, or to $v_{i-1}$ and $w_1$. In both cases we obtain a contradiction
as $B$ would contain either the $K_3 = \{v,v_i,v_{i+1}\}$ or the $K_3 = \{v,v_{i-2},v_{i-1}\}$ which is not possible if $B$ is a biclique. We can conclude then that there are no edges in common 
between the bicliques $\widetilde{B}^1_j$ and $\widetilde{B}'_j$, for all $j=1,\ldots,m$. Now let $\widetilde{b}^1_j$ be the vertices in $KB^2_e(G)$ corresponding to the bicliques $\widetilde{B}^1_j$ of $KB(G)$,
for $j=1,\ldots,m$ respectively. Then, these vertices form a $K_m$ in $KB^2_e(G)$ and they are only adjacent to the vertices $\widetilde{b}_{i-2}$ and $\widetilde{b}_{i-1}$ of the cycle $C''_n$. 
See Fig~\ref{divergeproof6}.

\begin{figure}[ht!]	
  \centering	
  \includegraphics[scale=.34]{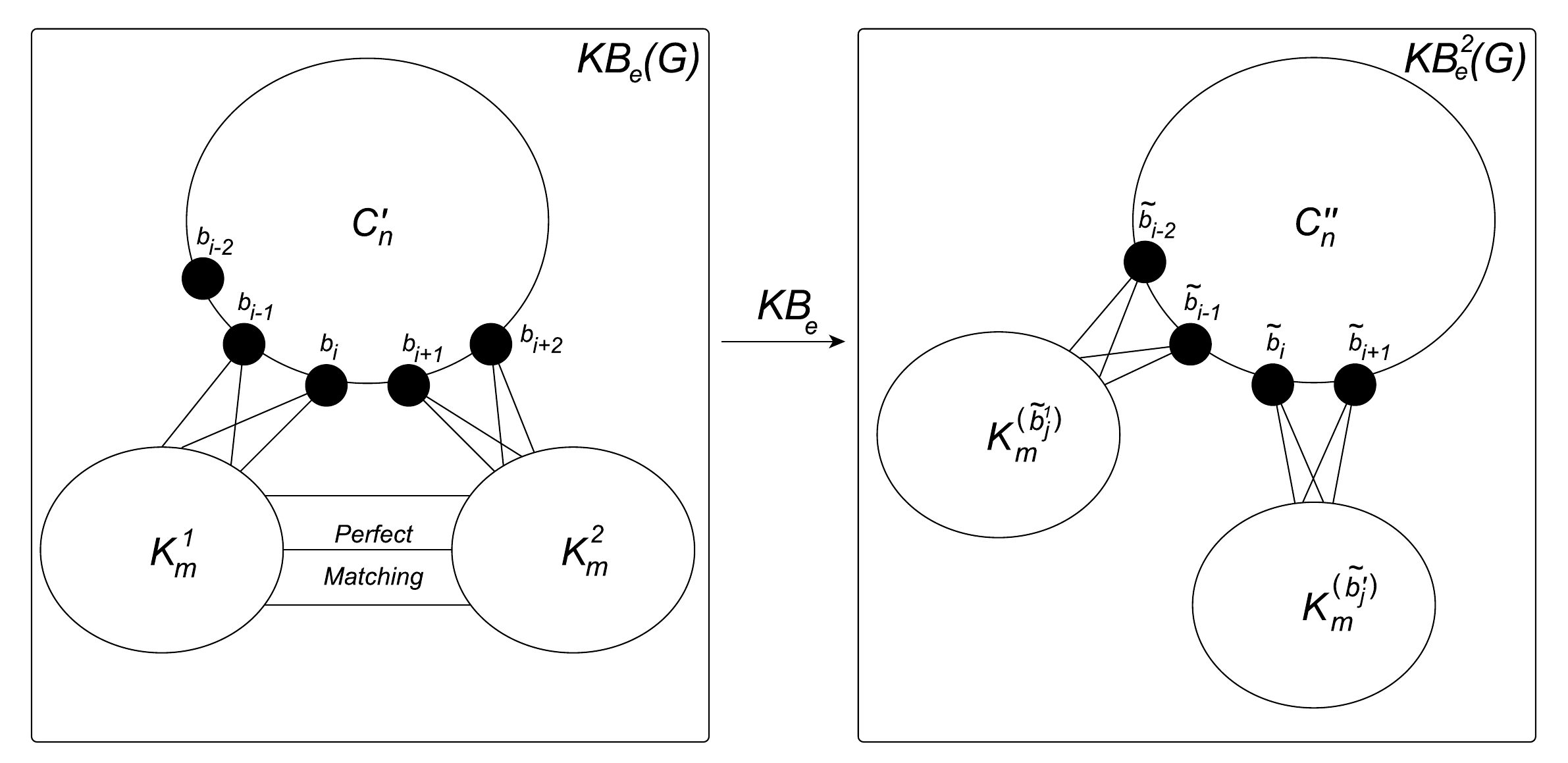} 	
  \caption{\textbf{Case B}: Second iteration of the operator $KB_e$.}
  \label{divergeproof6}	
\end{figure}

Now, let $\beta_t$, $t=0,\ldots,n-1$, be bicliques of $KB^2_e(G)$ that contain the vertices $\{\widetilde{b}_{t-1},\widetilde{b}_t,\widetilde{b}_{t+1}\}$ $(mod \ n)$, respectively, and
$\widetilde{\beta}_t$, $t=0,\ldots,n-1$, the corresponding vertices in $KB^3_e(G)$ to the bicliques $\beta_t$. By Proposition~\ref{condsylvain}, 
$C'''_n = \widetilde{\beta}_0\widetilde{\beta}_1\ldots \widetilde{\beta}_{n-1}$ is an induced cycle in $KB^3_e(G)$ with good neighbors. To finish, consider the following two families of bicliques: 
$\beta^1 = \{\beta^1_j : \{\widetilde{b}^1_j,\widetilde{b}_{i-1},\widetilde{b}_i\} \subseteq \beta^1_j, j=1,\ldots,m\}$ and 
$\beta^2 = \{\beta^2_j : \{\widetilde{b}'_j,\widetilde{b}_{i-1},\widetilde{b}_i\} \subseteq \beta^2_j, j=1,\ldots,m\}$. Clearly, all these $2m$ bicliques are different as there are no edges in common between the bicliques $\widetilde{B}^1_j$ and $\widetilde{B}'_j$, for all $j=1,\ldots,m$, and moreover, they are different to the bicliques $\beta_t$ for $t=0,\ldots,n-1$ as
$C''_n$ has good neighbors. Since all these $2m$ bicliques contain the edge $\widetilde{b}_{i-1}\widetilde{b}_i$, then if $\widetilde{\beta}^1_j$ and $\widetilde{\beta}^2_j$, $j=1,\ldots,m$, are the 
corresponding vertices in $KB^3_e(G)$ to the bicliques $\beta^1_j$ and $\beta^2_j$, respectively, we have that in $KB^3_e(G)$, 
$\{\widetilde{\beta}^1_1,\ldots,\widetilde{\beta}^1_m,\widetilde{\beta}^2_1,\ldots,\widetilde{\beta}^2_m\}$ is a complete graph on $2m$ vertices such 
that, as $C'''_n$ has good neighbors, every vertex of this $K_{2m}$ is only adjacent to $\widetilde{\beta}_i$ and to $\widetilde{\beta}_{i-1}$ on the cycle $C'''_n$. 
That is, $KB^3_e(G)$ contains an induced $(n,2m)-necklace$ such that its cycle $C'''_n$ has good neighbors. See Fig~\ref{divergeproof8}.
\end{proof}

\begin{figure}[ht!]	
  \centering	
  \includegraphics[scale=.36]{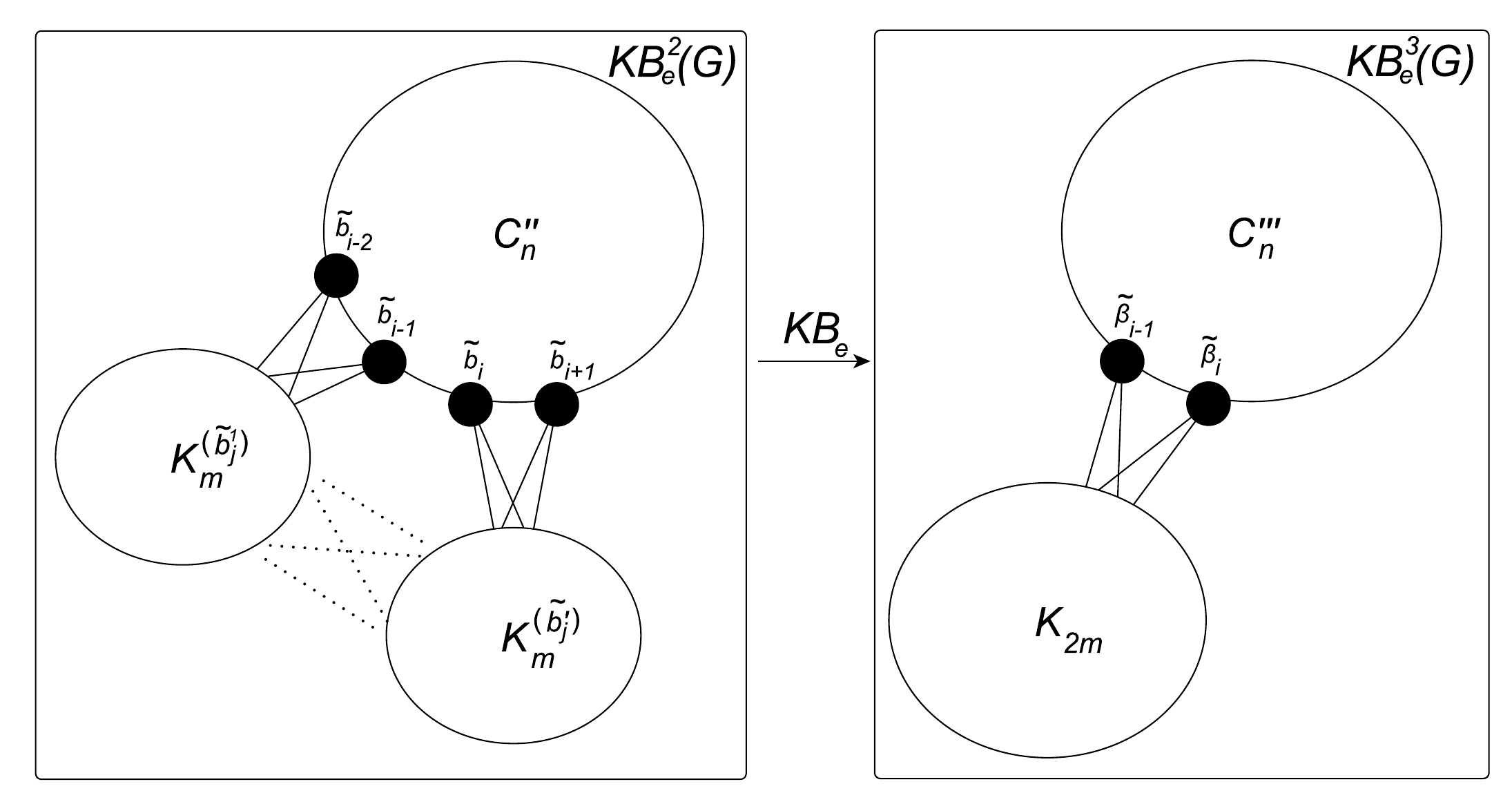} 	
  \caption{\textbf{Case B}: Third iteration of the operator $KB_e$.}
  \label{divergeproof8}	
\end{figure}

As a corollary, we obtain the following divergence theorem.

\begin{theorem}\label{neckdiverge}
Let $G$ be a graph that contains an induced $(n,m)-necklace$, $n \geq 5$, $m \geq 1$, such that its cycle has good neighbors. Then $G$ diverges under the operator $KB_e$. 
\end{theorem}
\begin{proof}
Applying Theorem~\ref{neck} several times, we obtain that either $KB^2_e(G)$ or $KB^3_e(G)$ contains an induced $(n,m')-necklace$, and $m'>m$.
Then $KB^4_e(G)$, $KB^5_e(G)$ or $KB^6_e(G)$ contains an induced $(n,m'')-necklace$, and $m''>m'$, etc., 
all having its cycles with good neighbors. Therefore, $G$ is divergent under the operator $KB_e$ as $\lim_{k \rightarrow \infty}|V(KB^k_e(G))|=\infty$.
\end{proof}

To finish the section, we obtain a second corollary.

\begin{corollary}
Let $G$ be a graph and let $C_n$ be an induced cycle of length $n \geq 5$ with good neighbors. If there is a vertex $v \in V(G) - C_n$ such that 
$N(v) \cap C_n$ has at least one edge and not all $C_n$, then $G$ diverges under the operator $KB_e$.
\end{corollary}
\begin{proof}
Just observe that $KB_e(G)$ satisfies conditions of Theorem~\ref{neckdiverge}.
\end{proof}

\section{Burgeon graphs}

In this section we will study the iterated edge-biclique graph of burgeon graphs and its relationship with the iterated line graph.

\begin{definition}
Let $G$ be a graph. We define the burgeon graph of $G$, denoted by $B(G)$, as the graph obtained by replacing each vertex $v$ of $G$ by a clique $C_v$ of $d(v)$ vertices, such that 
each vertex of the clique $C_v$ is only adjacent (to the outside of $C_v$) to exactly one vertex of another clique $C_u$ if and only if $u$ and $v$ are adjacent in $G$. See Fig~\ref{burgeon}. 
\end{definition}

\begin{figure}[ht!]	
  \centering	
  \includegraphics[scale=.36]{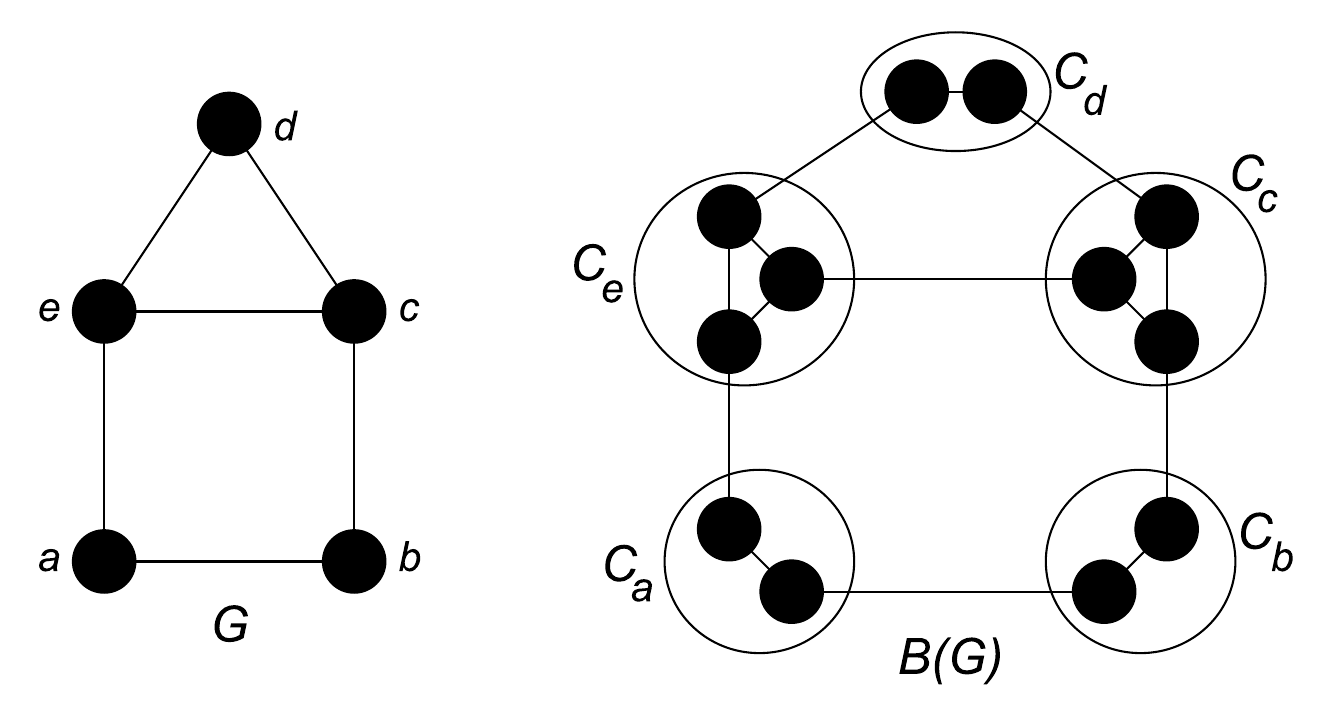} 	
  \caption{Graph $G$ and the construction of $B(G)$.}
  \label{burgeon}	
\end{figure}

Recall the definition of the line graph of a graph $G$, denoted by $L(G)$, as the intersection graph of the edges of $G$, that is, $L(G)$ has one vertex for each edge of $G$ and
two vertices $v,w$ in $L(G)$ are adjacent if their corresponding edges in $G$ have a common endpoint.
Next theorem shows the connection between the three operators $KB_e$, $B$ and $L$.

\begin{theorem}\label{burgues}
Let $G$ be a graph on $n \geq 2$ vertices. Then $KB_e(B(G)) = B(L(G))$.
\end{theorem}
\begin{proof}
Observe first that in $B(G)$ we have two types of edges. Edges of type I will be the edges inside the cliques and edges of type II will be the edges joining different cliques (these are in 
one-to-one correspondence with the edges of $G$). Now, as $B(G)$ has no induced $C_4$ and there is at most one edge of type II between each pair of cliques, we have that all bicliques of $B(G)$ are
isomorphic to $K_{1,2}$. Moreover, each biclique is formed with an edge of type I and an edge of type II sharing a common vertex.
Consider now an edge $e=vw \in G$ and its corresponding edge $e_B=v_Bw_B \in B(G)$, with $v_B \in C_v$ and $w_B \in C_w$. Note that $e_B$ is of type II. The edge $e_B$ belongs to 
exactly $d(v)-1+d(w)-1 = d(v)+d(w)-2$ bicliques in $B(G)$. First $d(v)-1$ bicliques are formed with the edge $e_B$ and each choice of an edge in $C_v$ having $v_B$ as an endpoint, while
the other $d(w)-1$ bicliques, with the edge $e_B$ and each choice of an edge in $C_w$ having $w_B$ as an endpoint. 
Since all these bicliques have the edge $e_B$ in common, their corresponding vertices in $KB_e(B(G))$ form a clique of size $d(v)+d(w)-2$. For each edge $e_B \in B(G)$ of type II, 
call $C_{e_B}$ this clique in $KB_e(B(G))$. Finally, observe that there is exactly one edge between two cliques $C_{e_B},C_{e'_B}$ of $KB_e(B(G))$ if and only if there are
two bicliques in $B(G)$ containing $e_B$ and $e'_B$ respectively, and a common edge of type I. That is, $e$ and $e'$ are adjacent in $G$.

Now, in $L(G)$, each vertex, say $e_L$ (that corresponds to an edge $e=vw$ of $G$), is adjacent to $d(v)+d(w)-2$ other vertices in $L(G)$. Thus, each vertex of $L$ will form a
clique, say $C_{e_L}$, of size $d(v)+d(w)-2$ vertices in $B(L(G))$. Finally, there is exactly one edge between two cliques $C_{e_L},C_{e'_L}$ of $B(L(G))$ if and only if the
vertices $e_L$ and $e'_L$ are adjacent in $L(G)$. That is, $e$ and $e'$ are adjacent in $G$.

We conclude therefore that $KB_e(B(G)) = B(L(G))$ as desired (see Figure~\ref{transformation}).
\end{proof}

\begin{figure}[ht!]	
  \centering	
  \includegraphics[scale=.36]{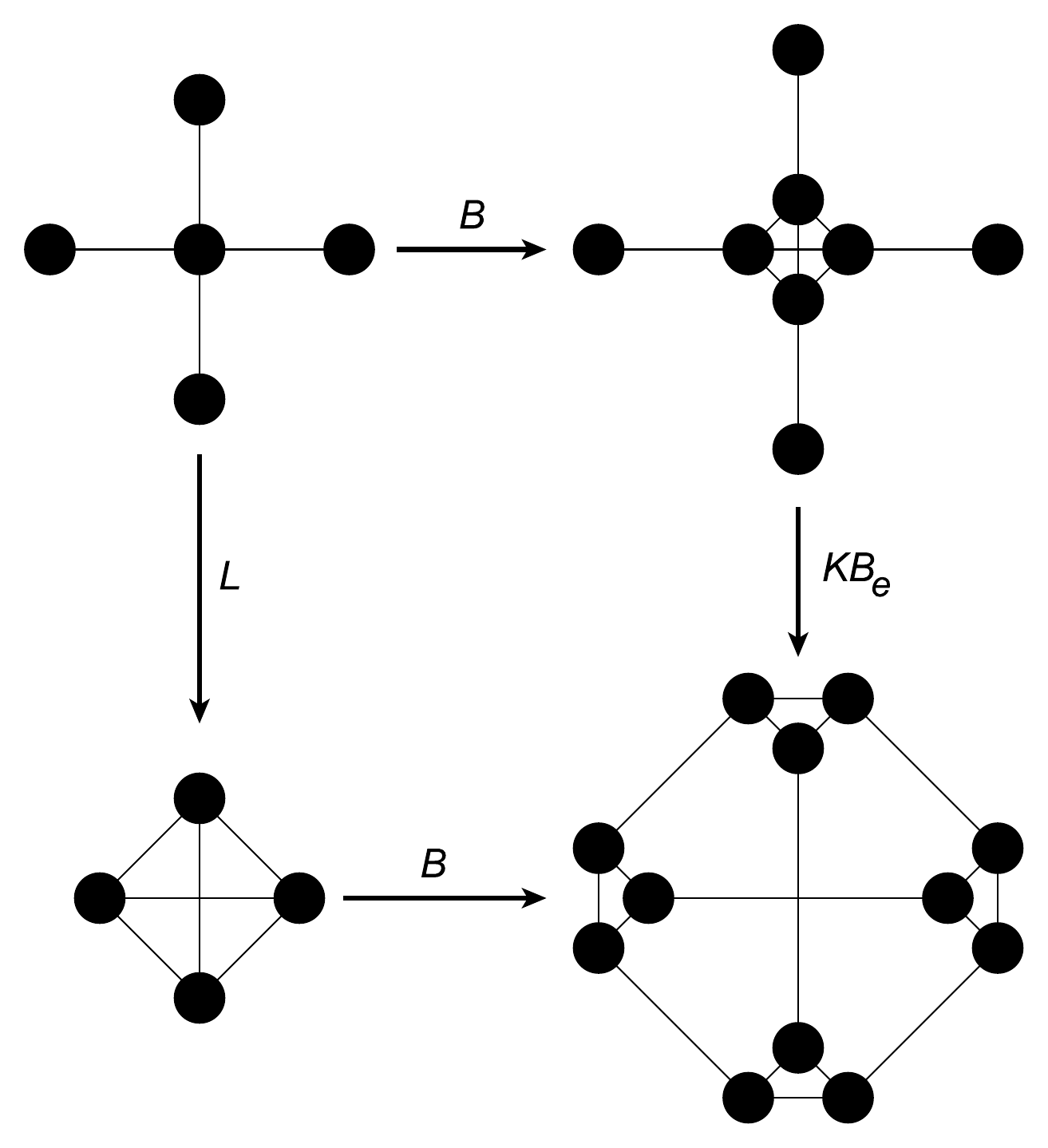} 	
  \caption{Example of the relationship of Theorem~\ref{burgues}.}
  \label{transformation}	
\end{figure}

As a corollary, we can characterize the behavior of burgeon graphs under the $KB_e$ operator.

\begin{corollary}\label{coroburgues}
Let $G$ be a graph. $B(G)$ is divergent under the $KB_e$ operator if and only if $G$ is not a cycle, a path or a $K_{1,3}$.
\end{corollary}
\begin{proof}
Theorem~\ref{burgues} implies that $KB_e^n(B(G)) = B(L^n(G))$. Also, we know by~\cite{linegraph} that $G$ diverges under the $L$ operator if and only if $G$ is not a cycle, a path, or a $K_{1,3}$. 
Combining both last statements along with the fact that $B(G)$ has at least as many vertices as $G$ (for $|V(G)| \neq 2$), the result holds.
\end{proof}

\begin{corollary}\label{coroburgues2}
Let $G$ be a graph. $B(G)$ is convergent under the $KB_e$ operator if and only if $G$ is a cycle, a path or a $K_{1,3}$. Moreover, it converges to itself, to the empty graph or to $C_6$, 
respectively.
\end{corollary}

Last corollary can be stated only in terms of burgeon graphs applying the $B$ operator as follows.

\begin{corollary}\label{coroburgues2-2}
Let $G=B(H)$ for some graph $H$. $G$ is convergent under the $KB_e$ operator if and only if $G$ is a cycle, a path or the net graph (see Fig~\ref{net}). Moreover, it converges to itself, to the empty graph or to $C_6$, respectively.
\end{corollary}

\begin{figure}[ht!]	
  \centering	
  \includegraphics[scale=.36]{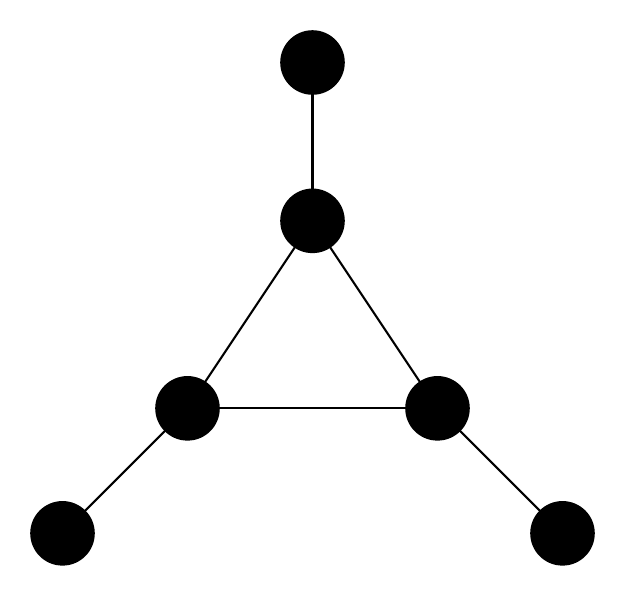} 	
  \caption{The \textit{net} graph.}
  \label{net}
\end{figure}

Note that one can verify in polynomial time if given a graph $G$, there exists some graph $H$ such that $G=B(H)$. Moreover, since checking if $G$ is a cycle, a path or the \textit{net} graph can 
also be done in polynomial time, we can conclude that deciding the behavoir of a burgeon graph under the $KB_e$ operator is polynomial as well.

We finish the section with the following result.

\begin{proposition}\label{propburguesneck}
Let $G=B(H)$ for some graph $H$. Then $KB_e(G)$ is a cycle, a path or it contains an induced $(n,m)-necklace$, $n\geq 6$ and $m\geq 1$, with good neighbors. 
\end{proposition}
\begin{proof}
By previous results, $G$ is either divergent or convergent under the $KB_e$ operator, therefore if it is convergent, then $G$ is a cycle, a path or the \textit{net} graph, thus
$KB_e(G)$ is a cycle, a (shorter) path or a $C_6$, respectively. Now, if it is divergent, then $G$ is not a cycle, a path or the \textit{net} graph, 
therefore $H$ is not a cycle, a path or a $K_{1,3}$. This implies that $H$ contains the \textit{paw} graph, the \textit{chair} graph (see Fig~\ref{paw-chair}) or a $K_{1,4}$, not 
necessarily induced. We will show that $B(L(H))$ contains an induced $(n,m)-necklace$, $n\geq 6$ and $m\geq 1$, with good neighbors, then by Theorem~\ref{burgues}, $B(L(H)) = KB_e(B(H)) = KB_e(G)$ 
contains it as well. 
We will also use the following remark; given a graph $X$, and $X'$ a subgraph of $X$ not necessarily induced, then $L(X')$ and $B(X')$ are induced subgraphs of $L(X)$ and $B(X)$, respectively.

Observe now that $L(paw)=diamond$, $L(K_{1,4})=K_4$ and $L(chair)=paw$, and $B(diamond)$, $B(K_4)$ and $B(paw)$, contain an induced $(n,m)-necklace$, $n\geq 6$ and $m\geq 1$, therefore
following the remark, $B(L(H)) = KB_e(G)$ also contains an induced $(n,m)-necklace$, $n\geq 6$ and $m\geq 1$.

Note that the induced $(n,m)-necklace$, $n\geq 6$ and $m\geq 1$, in $KB_e(G)$ always have good neighbors, since the operator $B$ applied to any graph, never contains an induced $C_4$.

\end{proof}

\begin{figure}[ht!]	
  \centering	
  \includegraphics[scale=.36]{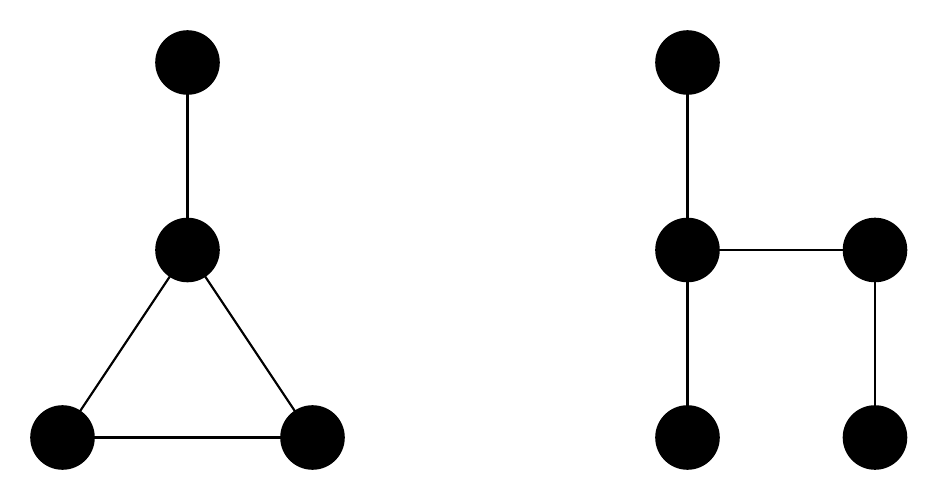} 	
  \caption{The \textit{paw} and the \textit{chair} graphs.}
  \label{paw-chair}
\end{figure}

\section{Open problems}

We propose the following conjectures.

\begin{conjecture}\label{comportamiento}
A graph $G$ is either divergent or convergent under the $KB_e$ operator but never periodic (with period bigger than $1$).
\end{conjecture}

\begin{conjecture}\label{self}
$G = KB_e(G)$ if and only if $G=\overline{C_7}$, $G=G_9$ or $G$ has girth at least five and has no vertices of degree one (see Fig.~\ref{c7-g9}).
\end{conjecture}

\begin{figure}[ht!]	
  \centering	
  \includegraphics[scale=.25]{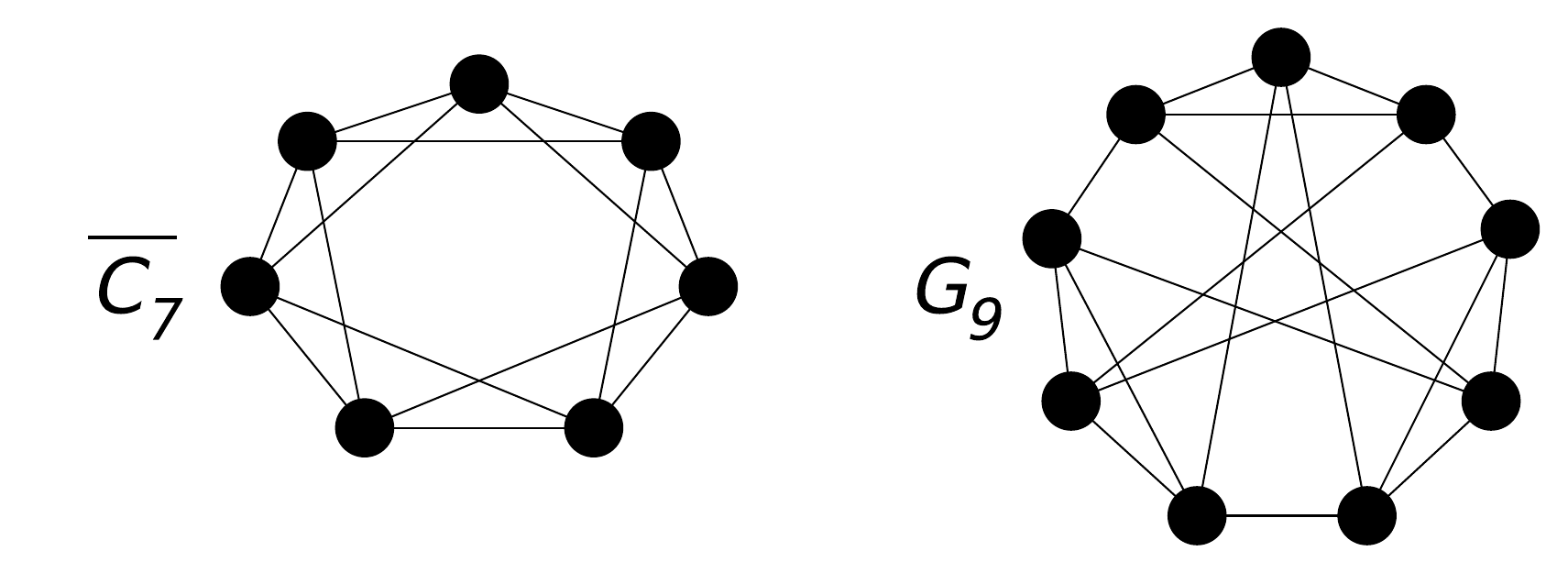} 	
  \caption{Graphs $\overline{C_7}$ and $G_9$ satisfying $KB_e(G) = G$ with girth less than five.}
  \label{c7-g9}	
\end{figure}

Note that Corollary~\ref{conv2} together with the fact that $KB_e(\overline{C_7}) = \overline{C_7}$, $KB_e(G_9) = G_9$ prove the ``only if'' part of Conjecture~\ref{self}.

\begin{conjecture}\label{computable}
It is computable to decide if a graph diverges or converges under the operator $KB_e$. 
\end{conjecture}

Despite that it seems that a small number of graphs contain an induced $(n,m)-necklace$, $n\geq 5$, $m\geq 1$, we believe that all divergent graphs will contain one in some iteration
under the operator $KB_e$.
We propose therefore the following conjecture.

\begin{conjecture}\label{caracdiverge}
A graph $G$ is divergent under the operator $KB_e$ if and only if there exists some $k$ such that $KB^k_e(G)$ contains an induced $(n,m)-necklace$, $n\geq 5$, $m\geq 1$, with its cycle having good neighbors.
\end{conjecture}

Clearly Theorem~\ref{neckdiverge} proves the ``only if'' part of Conjecture~\ref{caracdiverge} and moreover, the ``if'' part along with Conjecture~\ref{comportamiento} imply Conjecture~\ref{computable}. Note that all these conjectures are true for burgeon graphs.

\bibliography{biblio}

\begin{thebibliography}{10}

\bibitem{Alc'onFariaFigueiredoGutierrez2006}
L.~Alc\'on, L.~Faria, C.~M.~H. de~Figueiredo, and M.~Gutierrez.
\newblock The complexity of clique graph recognition.
\newblock {\em Theoret. Comput. Sci.}, 410(21-23):2072--2083, 2009.

\bibitem{Atluri}
G.~Atluri, J.~Bellay, G.~Pandey, C.~Myers, and V.~Kumar.
\newblock Discovering coherent value bicliques in genetic interaction data.
\newblock In {\em Proceedings of 9th International Workshop on Data Mining in
  Bioinformatics (BIOKDD'10)}, 2000.

\bibitem{BandeltPrisnerJCTSB1991}
H.-J. Bandelt and E.~Prisner.
\newblock Clique graphs and {H}elly graphs.
\newblock {\em J. Combin. Theory Ser. B}, 51(1):34--45, 1991.

\bibitem{BoothLuekerJCSS1976}
K.~Booth and G.~Lueker.
\newblock Testing for the consecutive ones property, interval graphs, and graph
  planarity using {$PQ$}-tree algorithms.
\newblock {\em J. Comput. System Sci.}, 13(3):335--379, 1976.

\bibitem{BrandstadtLeSpinrad1999}
A.~Brandst{\"a}dt, V.~Le, and J.~P. Spinrad.
\newblock {\em Graph {C}lasses: a {S}urvey}.
\newblock SIAM Monographs on Discrete Mathematics and Applications. Society for
  Industrial and Applied Mathematics (SIAM), Philadelphia, PA, 1999.

\bibitem{Bu01052003}
D.~Bu, Y.~Zhao, L.~Cai, H.~Xue, X.~Zhu, H.~Lu, J.~Zhang, S.~Sun, L.~Ling,
  N.~Zhang, G.~Li, and R.~Chen.
\newblock Topological structure analysis of the protein-protein interaction
  network in budding yeast.
\newblock {\em Nucleic Acids Research}, 31(9):2443--2450, 2003.

\bibitem{MelloMorganaLiveraniDAM2006}
C.~P. de~Mello, A.~Morgana, and M.~Liverani.
\newblock The clique operator on graphs with few {$P\sb 4$}'s.
\newblock {\em Discrete Appl. Math.}, 154(3):485--492, 2006.

\bibitem{EscalanteAMSUH1973}
F.~Escalante.
\newblock \"{U}ber iterierte {C}lique-{G}raphen.
\newblock {\em Abh. Math. Sem. Univ. Hamburg}, 39:59--68, 1973.

\bibitem{burgeon1}
O.~Favaron.
\newblock Irredundance in inflated graphs.
\newblock {\em J. Graph Theory}, 28(2):97--104, 1998.

\bibitem{burgeon3}
O.~Favaron.
\newblock Inflated graphs with equal independence number and upper irredundance
  number.
\newblock {\em Discrete Mathematics}, 236(1):81--94, 2001.

\bibitem{Frias-ArmentaNeumann-LaraPizanaDM2004}
M.~E. Fr{\'i}as-Armenta, V.~Neumann-Lara, and M.~A. Piza{\~n}a.
\newblock Dismantlings and iterated clique graphs.
\newblock {\em Discrete Math.}, 282(1-3):263--265, 2004.

\bibitem{FulkersonGrossPJM1965}
D.~R. Fulkerson and O.~A. Gross.
\newblock Incidence matrices and interval graphs.
\newblock {\em Pacific J. Math.}, 15:835--855, 1965.

\bibitem{GavrilJCTSB1974}
F.~Gavril.
\newblock The intersection graphs of subtrees in trees are exactly the chordal
  graphs.
\newblock {\em J. Combinatorial Theory Ser. B}, 16:47--56, 1974.

\bibitem{algorapide}
M.~Groshaus, A.~L. Guedes, and L.~Montero.
\newblock Almost every graph is divergent under the biclique operator.
\newblock {\em Discrete Appl. Math.}, 201:130 -- 140, 2016.

\bibitem{pavol}
M.~Groshaus, P.~Hell, and J.~Stacho.
\newblock On edge-sets of bicliques in graphs.
\newblock {\em Discrete Appl. Math.}, 160(18):2698 -- 2708, 2012.

\bibitem{marinayo}
M.~Groshaus and L.~Montero.
\newblock On the iterated biclique operator.
\newblock {\em J. Graph Theory}, 73(2):181--190, 2013.

\bibitem{arxivdist}
M.~Groshaus and L.~Montero.
\newblock Structural properties of biclique graphs and the distance formula.
\newblock {\em CoRR}, abs/1708.09686v5, 2021.

\bibitem{GroshausSzwarcfiterJGT2010}
M.~Groshaus and J.~L. Szwarcfiter.
\newblock Biclique graphs and biclique matrices.
\newblock {\em J. Graph Theory}, 63(1):1--16, 2010.

\bibitem{tesismarina}
M.~E. Groshaus.
\newblock {\em Bicliques, cliques, neighborhoods y la propiedad de Helly}.
\newblock PhD thesis, Universidad de Buenos Aires, 2006.

\bibitem{Haemers200156}
W.~H. Haemers.
\newblock Bicliques and eigenvalues.
\newblock {\em J. Combinatorial Theory Ser. B}, 82(1):56 -- 66, 2001.

\bibitem{HamelinkJCT1968}
R.~C. Hamelink.
\newblock A partial characterization of clique graphs.
\newblock {\em J. Combinatorial Theory}, 5:192--197, 1968.

\bibitem{HedetniemiSlater1972}
S.~T. Hedetniemi and P.~J. Slater.
\newblock Line graphs of triangleless graphs and iterated clique graphs.
\newblock In {\em Graph theory and applications ({P}roc. {C}onf., {W}estern
  {M}ichigan {U}niv., {K}alamazoo, {M}ich., 1972; dedicated to the memory of
  {J}. {W}. {T}. {Y}oungs)}, pages 139--147. Lecture Notes in Math., Vol. 303.
  Springer, Berlin, 1972.

\bibitem{burgeon2}
M.~A. Henning and A.~P. Kazemi.
\newblock Total domination in inflated graphs.
\newblock {\em Discrete Appl. Math.}, 160(1-2):164--169, 2012.

\bibitem{Kumar}
R.~Kumar, P.~Raghavan, S.~Rajagopalan, and A.~Tomkins.
\newblock Trawling the web for emerging cyber-communities.
\newblock In {\em Proceeding of the 8th international conference on World Wide
  Web, pages 1481--1493, 1999.}, 2000.

\bibitem{LarrionMelloMorganaNeumann-LaraPizanaDM2004}
F.~Larri{\'o}n, C.~P. de~Mello, A.~Morgana, V.~Neumann-Lara, and M.~A.
  Piza{\~n}a.
\newblock The clique operator on cographs and serial graphs.
\newblock {\em Discrete Math.}, 282(1-3):183--191, 2004.

\bibitem{LarrionNeumann-LaraGC1997}
F.~Larri{\'o}n and V.~Neumann-Lara.
\newblock A family of clique divergent graphs with linear growth.
\newblock {\em Graphs Combin.}, 13(3):263--266, 1997.

\bibitem{LarrionNeumann-LaraDM1999}
F.~Larri{\'o}n and V.~Neumann-Lara.
\newblock Clique divergent graphs with unbounded sequence of diameters.
\newblock {\em Discrete Math.}, 197/198:491--501, 1999.
\newblock 16th British Combinatorial Conference (London, 1997).

\bibitem{LarrionNeumann-LaraDM2000}
F.~Larri{\'o}n and V.~Neumann-Lara.
\newblock Locally {$C\sb 6$} graphs are clique divergent.
\newblock {\em Discrete Math.}, 215(1-3):159--170, 2000.

\bibitem{LarrionNeumann-LaraPizanaDM2002}
F.~Larri{\'o}n, V.~Neumann-Lara, and M.~A. Piza{\~n}a.
\newblock Whitney triangulations, local girth and iterated clique graphs.
\newblock {\em Discrete Math.}, 258(1-3):123--135, 2002.

\bibitem{LarrionPizanaVillarroel-FloresDM2008}
F.~Larri{\'o}n, M.~A. Piza{\~n}a, and R.~Villarroel-Flores.
\newblock Equivariant collapses and the homotopy type of iterated clique
  graphs.
\newblock {\em Discrete Math.}, 308:3199--3207, 2008.

\bibitem{LehotJA1974}
P.~G.~H. Lehot.
\newblock An optimal algorithm to detect a line graph and output its root
  graph.
\newblock {\em J. ACM}, 21(4):569--575, 1974.

\bibitem{blablamec}
G.~Liu, K.~Sim, and J.~Li.
\newblock Efficient mining of large maximal bicliques.
\newblock In {\em Proceedings of the 8th International Conference on Data
  Warehousing and Knowledge Discovery}, DaWaK'06, page 437–448, Berlin,
  Heidelberg, 2006. Springer-Verlag.

\bibitem{McKeeMcMorris1999}
T.~A. McKee and F.~R. McMorris.
\newblock {\em Topics in {I}ntersection {G}raph {T}heory}.
\newblock SIAM Monographs on Discrete Mathematics and Applications. Society for
  Industrial and Applied Mathematics (SIAM), Philadelphia, PA, 1999.

\bibitem{tesislea}
L.~Montero.
\newblock Convergencia y divergencia del grafo biclique iterado.
\newblock Master's thesis, Departamento de Computaci\'on, Facultad de Ciencias
  Exactas y Naturales, Universidad de Buenos Aires, 2008.

\bibitem{Niranjan}
N.~Nagarajan and C.~Kingsford.
\newblock Uncovering genomic reassortments among influenza strains by
  enumerating maximal bicliques.
\newblock {\em 2012 IEEE International Conference on Bioinformatics and
  Biomedicine}, 0:223--230, 2008.

\bibitem{Neumann1981}
V.~Neumann~Lara.
\newblock Clique divergence in graphs.
\newblock In {\em Algebraic methods in graph theory, {V}ol. {I}, {II}
  ({S}zeged, 1978)}, volume~25 of {\em Colloq. Math. Soc. J\'anos Bolyai},
  pages 563--569. North-Holland, Amsterdam, 1981.

\bibitem{PizanaDM2003}
M.~A. Piza{\~n}a.
\newblock The icosahedron is clique divergent.
\newblock {\em Discrete Math.}, 262(1-3):229--239, 2003.

\bibitem{PrisnerC2000}
E.~Prisner.
\newblock Bicliques in graphs i: Bounds on their number.
\newblock {\em Combinatorica}, 20(1):109--117, 2000.

\bibitem{RobertsSpencerJCTSB1971}
F.~S. Roberts and J.~H. Spencer.
\newblock A characterization of clique graphs.
\newblock {\em J. Combinatorial Theory Ser. B}, 10:102--108, 1971.

\bibitem{Szpilrajn-MarczewskiFM1945}
E.~Szpilrajn-Marczewski.
\newblock Sur deux propri\'et\'es des classes d'ensembles.
\newblock {\em Fund. Math.}, 33:303--307, 1945.

\bibitem{linegraph}
A.~C.~M. van Rooij and H.~S. Wilf.
\newblock The interchange graph of a finite graph.
\newblock {\em Acta Math. Acad. Sci. Hungar.}, 16:263--269, 1965.

\end{thebibliography}

\end{document}